\documentclass[journal]{IEEEtran}
\usepackage{cite}
\usepackage{amsmath,amssymb,amsfonts}
\usepackage{algorithmic}
\usepackage{graphicx}
\usepackage{textcomp}
\usepackage{mathtools}				
\usepackage{bm}  					
\usepackage[ruled,linesnumbered]{algorithm2e}
\usepackage{subfigure}
\usepackage{xcolor}
\usepackage{color}
\usepackage{stfloats}
\usepackage[hidelinks]{hyperref}
\usepackage{endnotes}
\usepackage{url}

\usepackage{amsthm}

\newtheorem{theorem}{Theorem}
\newtheorem{lemma}{Lemma}

\newtheorem{definition}{Definition}

\allowdisplaybreaks[4]

\newcommand{\T}{\top}

\newcommand{\tabincell}[2]{\begin{tabular}{@{}#1@{}}#2\end{tabular}}

\usepackage{multirow}
\usepackage{array}
\usepackage{booktabs}
\graphicspath{{figure/}}

\usepackage{xcolor}

\begin{document}

	\title{Autonomous Tracking and State Estimation with Generalised Group Lasso}
	\author{Rui~Gao,
		Simo~S\"arkk\"a,
		Rub\'en~Claveria-Vega,	
		and Simon~Godsill
		\thanks{R.~Gao and S.~S\"arkk\"a are with the Department of Electrical Engineering and Automation, Aalto University, Espoo, 02150 Finland (e-mail:rui.gao@aalto.fi; simo.sarkka@aalto.fi).
			R.~Claveria-Vega and S.~Godsill are with the Department of Engineering, University of Cambridge, CB2 1PZ, UK (e-mail:rmc83@cam.ac.uk; sjg@eng.cam.ac.uk).
			This work was supported by Academy of Finland.
		}
	}

	\markboth{Journal of \LaTeX\ Class Files,~Vol.~14, No.~8, August~2015}
	{Shell \MakeLowercase{\textit{et al.}}: Bare Demo of IEEEtran.cls for IEEE Journals}

	\maketitle

	\begin{abstract}
		We address the problem of autonomous tracking and state estimation for marine vessels, autonomous vehicles, and other dynamic signals under a (structured) sparsity assumption. The aim is to improve the tracking and estimation accuracy with respect to classical Bayesian filters and smoothers. We formulate the estimation problem as a dynamic generalised group Lasso problem and develop a class of smoothing-and-splitting methods to solve it. The Levenberg--Marquardt iterated extended Kalman smoother-based multi-block alternating direction method of multipliers (LM-IEKS-mADMM) algorithms are based on the alternating direction method of multipliers (ADMM) framework. 
		This leads to minimisation subproblems with an inherent structure to which three new augmented recursive smoothers are applied.
		Our methods can deal with large-scale problems without pre-processing for dimensionality reduction. Moreover, the methods allow one to solve nonsmooth nonconvex optimisation problems. We then prove that under mild conditions, the proposed methods converge to a stationary point of the optimisation problem. By simulated and real-data experiments including multi-sensor range measurement problems, marine vessel tracking, autonomous vehicle tracking, and audio signal restoration, we show the practical effectiveness of the proposed methods.
	\end{abstract}


	\IEEEpeerreviewmaketitle

	\section{Introduction}
	\IEEEPARstart{A}{utonomous} tracking and state estimation problems
	are active research topics with many real-world applications, including intelligent maritime navigation, autonomous vehicle tracking, and audio signal estimation \cite{simo2013Bayesian,Dornaika2009tracking,Rogez2014Monocular,Godsill1998Digital,Ahmad2016Intent}. The aim is to autonomously estimate and track the state (e.g., position, velocity, or direction) of the dynamic system using imperfect measurements \cite{Bar-Shalom+Li+Kirubarajan:2001}. A frequently used approach for autonomous tracking and estimation problems is based on Bayesian filtering and smoothing. When the target dynamics and observation models are linear and Gaussian, Kalman smoother (KS) \cite{simo2013Bayesian,RTS1965} provides the optimal Bayesian solution which coincides with the optimal minimum mean square error estimator in that case. In the case of nonlinear dynamic systems, the iterated extended Kalman smoother (IEKS) \cite{Bell1994smoother,Barfoot2017State,simo2020LMIEKS} makes use of local affine approximations by means of Taylor series for the nonlinear functions, and then iteratively carries out KS. Sigma-point based smoothing methods \cite{simo2008RTS,Garcia:2017} employ sigma-points to approximate the probability density of the states, which can preserve higher-order accuracy than IEKS. Random sampling-based filters such as particle filters \cite{Doucet+Godsill+Andrieu:2000,Doucet+Freitas+Gordon:2001,simo2013Bayesian,Seifzadeh2015particle} can be used to deal with nonlinear tracking situations involving potentially arbitrary nonlinearities, noise, and constraints.
	Although these trackers and estimators are capable of utilising the measurement information to obtain the estimates, they ignore sparsity dictated by physical attributes of dynamic systems.
	
	The motivation for our work comes from the following real-world applications. One significant application is \textit{marine vessel tracking} \cite{Ahmad2016Intent,Bar-Shalom+Li+Kirubarajan:2001}. Vessels are frequently pitching and rolling on the surface of the ocean, which can be modelled as sparsity in the process noise. Our methodology is also applicable to \textit{autonomous vehicle tracking}, which enables a vehicle to autonomously avoid obstacles and maintain safe distances to other vehicles. In presence of many sudden stops (i.e., velocities are zero), the tracking accuracy can be improved by employing sparsity \cite{Gao2019ieks}. Other examples of tracked targets include robots \cite{Barfoot2017State} and unmanned aerial vehicles \cite{Tisdale2009Autonomous}. Another practical application is \textit{audio signal restoration}, where typically only a few time-frequency elements are expected to be present, and thus, sparsity is an advisable assumption \cite{fevotte2007sparse}. 
	For example, Gabor synthesis representation with sparsity constraints has been proved to be suitable for audio restoration \cite{wolfe2004bayesian}.  
	Similar problems can also be found in electrocardiogram (ECG) signal analysis \cite{zhao2018spectro} and automatic music transcription \cite{Gao2019parameter}. Hence, computationally effective sparsity modelling methods are in demand.
	
	Since sparsity may improve the tracking and estimation performance, there is a growing literature that proposes sparse regularisers such as Lasso (i.e., least absolute shrinkage and selection operator, or $L_1$-regularisation) \cite{Tibshirani1996lasso,Meier2008groupLasso} or total variation (TV)~\cite{Rudin1992tv,Gao2018combined} for these applications. Existing methods for sparse tracking and estimation can be split in two broad categories: robust smoothing approaches and optimisation-based approaches. 
	The former approaches merge filtering and smoothing with $L_1$-regularisation. For instance, the modified recursive filter-based compressive sensing methods were developed in \cite{Carmi2010pseudo,Vaswani2008compressed,Recursive2016Recovery,Yoo2020Estimation}. An $L_1$-Laplace robust Kalman smoother was presented in \cite{Aravkin2011Laplace}. Using both sparsity and dynamics information, a sparse Bayesian learning framework was proposed in \cite{Shaughnessy2020Sparse}.
	The latter approaches formulate the whole tracking and state estimation problem as an $L_1$-penalised minimisation problem, and then apply iterative algorithms to solve the minimisation problem \cite{Aravkin2017Generalized,Charles2016Dynamic,Gao2019ieks,Simonetto2017Prediction,Farahmand2011sparsity,Gao2020sigma}. While these $L_1$-penalised estimators offer several benefits, they penalise individual elements of the state vector or process noise instead of groups of elements in them.	
	
	Recently, there have been important advances in the structured sparsity methodology for collision avoidance \cite{park2019transformation,Peng2020Output-Feedback}, path-following and tracking \cite{McCall2016Integral,Rout2020Sideslip-Compensated}, and 
	visual tracking \cite{Wu2017Exploiting, Bai2015learning}. In \cite{Du2016Discriminative}, a discriminative supervised hashing method was proposed for object tracking tasks. An adaptive elastic echo state network was developed for multivariate time series prediction in \cite{Xu2016Adaptive}. The work in \cite{Cheng2017Object} formulated a moving tracking problem with $L_2$-norm constraints and then introduced a temporal consistency dictionary learning algorithm. However, the methods lack strong convergence and performance guarantees, particularly when the objective becomes nonconvex. Moreover, relatively few methods exist for incorporating structured sparsity into autonomous tracking and state estimation problems. Taking these developments into consideration, the main goal here is to develop new efficient methods for regularised autonomous tracking and estimation problems, which allow for group Lasso type of sparseness assumptions on groups of state or process noise elements. 
	
	When we formulate a regularised autonomous tracking and state estimation problem as a generalised $L_2$-minimisation problem (also called as dynamic generalised group Lasso problem), the resulting problem is difficult to solve due to its nonsmoothness and/or nonconvexity. Splitting-based optimisation methods \cite{Wright2006Numerical,Boyd2011admm,Splitting2017book} such as multi-block alternating direction method of multipliers (ADMM) \cite{Boyd2011admm}, are methods that can tackle this kind of problems. One advantage of these methods is that they decompose the original problem into a sequence of easier subproblems. Although these methods can directly work on the original optimisation problem, such direct use ignores the inherent structure induced by the implied Markovian structure in the optimisation problem. In this paper, we propose a class of efficient smoothing-and-splitting methods that outperform the classical optimisation methods in terms of computational time due to the leveraging of the Markovian structure. 
	
	In this paper, we focus on autonomous tracking and state estimation problems with sparsity-inducing priors. Our first contribution is to provide a flexible formulation of dynamic generalised group Lasso problems arising in autonomous tracking and state estimation. Special cases of the formulation are Lasso, isotropic TV, anisotropic TV, fused Lasso, group Lasso, and sparse group Lasso. Meanwhile, the formulation can cope with sparsity on the process noise or the state in dynamic systems. Since the resulting optimisation problems are nonsmooth, possibly nonconvex, and large-dimensional, our second contribution is to provide a class of the smoothing-and-splitting methods to address them. We develop the new KS-mADMM, Gauss--Newton IEKS-mADMM (GN-IEKS-mADMM), and Levenberg--Marquardt IEKS-mADMM (LM-IEKS-mADMM) methods which use augmented recursive smoothers to solve the primal subproblems in the mADMM iterations. As a third contribution, we prove that under mild conditions, the proposed methods converge to a stationary point. Our fourth contribution is to apply the proposed methods to real-world applications of marine vessel tracking, autonomous vehicle tracking, and audio signal restoration.
	
	The rest of this paper is structured as follows. In Section~\ref{sec:formulation}, we formulate the sparse autonomous tracking and state estimation problem as a generalised $L_2$-minimisation problem. Particularly, we present a broad class of regulariser configurations parametrised by sets of matrices and vectors. We introduce the batch tracking and estimation methods in Section~\ref{sec:batch} and present three augmented recursive smoothing methods in Section~\ref{sec:proposed}. In Section~\ref{sec:convergence}, we establish the convergence. In Section~\ref{sec:results}, we report numerical results on simulated and real-life datasets. Section~\ref{sec:Conclusion} draws the concluding remarks.
	
	The notation is as follows. Matrices $\mathbf{X}$ and vectors $\mathbf{x}$ are indicated in boldface. $(\cdot)^{\T}$ represents transposition, and $(\cdot)^{-1}$ represents matrix inversion. The $\mathbf{R}$-weighted norm of $\mathbf{x}$ is denoted by $\|\mathbf{x}\|_{\mathbf{R}}= \sqrt{\mathbf{x}^\T \mathbf{R} \mathbf{x}}$. $\left \|\mathbf{x} \right \|_1 = \sum |x_i|$ denotes $L_1$-norm, and $\left \|\mathbf{x}\right \|_2 = \sqrt{\sum_i x_i^2}$ denotes $L_2$-norm. $\mathbf{X}_{g,t}$ is the $(g,t)$:th element of matrix $\mathbf{X}$, and $\mathbf{x}^{(k)}$ denotes the value of $\mathbf{x}$ at $k$:th iteration. 
	$\operatorname{vec}(\cdot)$ represents a vectorisation operator, $\operatorname{diag}(\cdot)$ represents a block diagonal matrix operator with the elements in its argument on the diagonal,  and $\mathbf{x}_{1:T} =\operatorname{vec}(\mathbf{x}_1,\ldots,\mathbf{x}_T)$. 
	$\partial \phi(\mathbf{x})$ denotes a sub-gradient of $\phi$. $\mathbf{J}_{\phi}$ is the Jacobian of ${\phi}(\mathbf{x})$. $\delta_{+}(\mathbf{A})$ denotes the smallest eigenvalue of $\mathbf{A}$. $p(\mathbf{x})$ denotes probability density function (pdf) of $\mathbf{x}$ and $\mathcal{N} (\mathbf{x}  \mid \mathbf{m}, \mathbf{P})$ denotes a Gaussian probability density function with mean $\mathbf{m}$ and covariance $\mathbf{P}$ evaluated at $\mathbf{x}$.

	\section{Problem Statement}
	\label{sec:formulation}
	
	Let $\mathbf{y}_{t} \in \mathbb{R}^{N_y}$ be a measurement of a dynamic system and $\mathbf{x}_{t} \in \mathbb{R}^{N_x}$ be an unknown state (sometimes called the source or signal). The state and measurement are related according to a dynamic state-space model of the form
	\begin{equation}\label{eq:model}
		\begin{split}
			\mathbf{x}_{t} =\mathbf{a}_t(\mathbf{x}_{t-1}) + \mathbf{q}_t,  \quad 
			\mathbf{y}_{t} =\mathbf{h}_t(\mathbf{x}_t) +\mathbf{r}_t, 
		\end{split}
	\end{equation}
	where $\mathbf{h}_t: \mathbb{R}^{N_x} \to \mathbb{R}^{N_y}$ and $\mathbf{a}_t: \mathbb{R}^{N_x} \to \mathbb{R}^{N_x}$ are the measurement and state transition functions, respectively, and $t = 1,\ldots,T$ is the time step number. The process and measurement noises $\mathbf{q}_t \sim \mathcal{N}(\mathbf{0},\mathbf{Q}_t)$ and $\mathbf{r}_t \sim \mathcal{N}(\mathbf{0},\mathbf{R}_t)$ are assumed to be zero-mean Gaussian with covariances $\mathbf{Q}_t$ and $\mathbf{R}_t$, respectively. The initial condition at $t=1$ is given by $\mathbf{x}_1\sim \mathcal{N}(\mathbf{m}_1,\mathbf{P}_1)$. A particular special case of \eqref{eq:model} is an affine Gaussian model by
	\begin{equation}\label{eq:linear_model}
		\begin{aligned}
			\mathbf{a}_{t}(\mathbf{x}_{t-1}) =\mathbf{A}_t\,\mathbf{x}_{t-1} + \mathbf{b}_t, \quad
			\mathbf{h}_{t}(\mathbf{x}_{t}) =\mathbf{H}_t\,\mathbf{x}_t +\mathbf{e}_t, 
		\end{aligned}
	\end{equation}
	where $\mathbf{A}_t  \in \mathbb{R}^{N_x \times N_x}$, $\mathbf{H}_t \in \mathbb{R}^{N_y \times N_x}$ are the transition and measurement matrices, and $\mathbf{b}_t$, $\mathbf{e}_t$ are bias terms.

	The goal here is to obtain the ``best estimate'' of $\mathbf{x}_{1:T}$ from imperfect measurements $\mathbf{y}_{1:T}$. 	
	For computing $\mathbf{x}_{1:T}$ with sparsity-inducing priors, we define a set of matrices $\left\{\mathbf{G}_{g,t}  \in \mathbb{R}^{P_g \times N_x} \mid g = 1, \ldots, N_g\right\}$, matrices $\mathbf{B}_t$, and vectors $\mathbf{d}_t$, for $t=1,\ldots,T$, and impose sparsity on the groups of elements of the state or the process noise. Mathematically, the problem of computing the state estimate $\mathbf{x}^\star_{1:{T}}$ is formulated as 
	\begin{equation}\label{eq:whole_optimisation}
		\begin{split}
			\begin{aligned}
				&\mathbf{x}^\star_{1:{T}} = \arg\min_{\mathbf{x}_{1:{T}}}
				\frac{1}{2} \sum_{t=1}^{T}  \left\| \mathbf{y}_t - \mathbf{h}_t(\mathbf{x}_t)\right\|_{\mathbf{R}_t^{-1}}^2  \\
				& \quad + \frac{1}{2} \sum_{t=2}^{T}  \left\| \mathbf{x}_t - \mathbf{a}_t(\mathbf{x}_{t-1})\right\|_{\mathbf{Q}_t^{-1}}^2  
				+\frac{1}{2} \left\| \mathbf{x}_1  -   \mathbf{m}_1 \right\|_{ \mathbf{P}_1^{-1} }^2 \\
				& \qquad +   \sum_{t=1}^{T}\sum_{g = 1}^{N_g} \mu \left \| \mathbf{G}_{g,t} \left(\mathbf{x}_t - \mathbf{B}_t \, \mathbf{x}_{t-1} - \mathbf{d}_t\right) \right \|_2, \\
			\end{aligned}
		\end{split}
	\end{equation}
	where $\mu > 0$ is a penalty parameter. 
	
	A merit of our formulation is its flexibility, because the selections of $\mathbf{G}_{g,t}$, $\mathbf{B}_{t}$, and $\mathbf{d}_t$ can be adjusted to represent different regularisers. With matrix $ \mathbf{G}_{g,t}$, the formulation~\eqref{eq:whole_optimisation} accommodates a large class of sparsity-promoting regularisers (e.g., Lasso, isotopic TV, anisotopic TV, fused Lasso, group Lasso, and sparse group Lasso). A list of such regularisers is reported in Table~\ref{tab:penalty}. Meanwhile, the formulation~\eqref{eq:whole_optimisation} also allows for putting sparsity assumptions on the state or the process noise by different selections of $\mathbf{B}_t$ and $\mathbf{d}_t$ (see Table~\ref{tab:selection_sparsity}).	
	\setlength\tabcolsep{2 pt}
	\begin{table}[!tbh]
		\caption{Examples of sparsity-promoting regularisers that are included in the present framework.}
		\begin{center}
			\begin{tabular}{|l|l|}
				\hline
				\textbf{Regularisation}	  &  $\mathbf{G}_{g,t} $ \textbf{descriptions} 
				\\ \hline
				$L_2$-regularisation   &  
				$\mathbf{G}_{g,t}$ is an identity matrix
				\\ \hline
				Lasso          &    \tabincell{l}{  $N_g = N_x$, $P_g = 1$ for all $g$, \\
					$\mathbf{G}_{g,t}$ has $1$ at $g$:th column and zeros otherwise.}
				\\  \hline  
				Isotopic TV    &  \tabincell{l}{ $N_g = 1$, $P_1 = N_x-1$ \\ 
					$ \mathbf{G}_{1,t}$ encodes a finite difference operator.}   
				\\  \hline
				Anisotopic TV   &  \tabincell{l}{$ \mathbf{G}_{g,t}$ encodes the $g$:th row of \\
					a finite difference operator. }   
				\\  \hline
				Fused Lasso  &   \tabincell{l}{$g = 1,\ldots, N_x$, $P_g = 1$ for all $g$,\\ 
					$\mathbf{G}_{g,t}$ has $1$ at $g$:th column and zeros otherwise; \\ 
					$g = N_x + 1,\ldots, N_g$,
					$ \mathbf{G}_{g,t}$ encodes the $g$:th row of \\
					a finite difference operator. } 
				\\ \hline
				Group Lasso  &   \tabincell{l}{ $ \mathbf{G}_{g,t}$ has $1$, corresponding to the selected elements \\ of $\mathbf{x}_{t}$ in the group and zeros otherwise. }   
				\\ \hline
				Sparse group Lasso & \tabincell{l}{$g = 1,\ldots, N_x$, $P_g = 1$,\\ 
					$\mathbf{G}_{g,t}$ has $1$ at $g$:th column and zeros otherwise;  \\      
					$g = N_x+1,\ldots, N_g$, \\
					$ \mathbf{G}_{g,t}$ has the same setting with group Lasso.}  
				\\ \hline
			\end{tabular}\label{tab:penalty}
		\end{center}
	\end{table}
	
	\setlength\tabcolsep{2pt}
	\begin{table}[!tbh]
		\caption{Flexible sparsity assumptions by selecting $\mathbf{B}_{t} $ and $\mathbf{d}_{t} $.}
		\begin{center}
			\begin{tabular}{|l|l|l|}
				\hline
				Dynamic systems  & Sparsity on:  & $\mathbf{B}_{t} $ and $\mathbf{d}_{t} $   \\ \hline
				\multirow{2}{*}{Affine Gaussian}    & state      &  
				Settings on $\mathbf{B}_t = \mathbf{0}$, $\mathbf{d}_t =  \mathbf{0}$                                                                               \\ \cline{2-3} 
				& process noise & 	$\mathbf{B}_t = \mathbf{A}_t$, $\mathbf{d}_t =  \mathbf{b}_t$                                                                                  \\ \hline
				\multirow{2}{*}{ Nonlinear} & state      & 			$\mathbf{B}_t = \mathbf{0}$, 
				$\mathbf{d}_t= \mathbf{0}$                                                                                  \\ \cline{2-3} 
				& process noise  &   \tabincell{l}{$\mathbf{B}_{t} \mathbf{x}_{t-1} \!+\! \mathbf{d}_t$ as the affine approximation \\
					of $\mathbf{a}_t(\mathbf{x}_{t-1})$ (see Section \ref{sec:nonlinear_ieks})}                                                                              \\ \hline
			\end{tabular}
			\label{tab:selection_sparsity}
		\end{center}
	\end{table}
	
	A simple, yet illustrative, example can be found in autonomous vehicle tracking. When there are stop-and-go points (e.g. vehicle stops) in the data, the zero-velocity and zero-angle values at those time points can be grouped together via the $L_2$-norm and $\mathbf{G}_{g,t}$. That means three elements can be forced to be equal to zero at the same time. Another application is in audio restoration, where the matrices $\mathbf{G}_{g,t}$ are defined so that only two elements of the state $\mathbf{x}_t$ -- corresponding to the real and imaginary parts of a synthesis coefficient -- are extracted at a time step. Thus, these pairs, which are associated with the same time-frequency basis functions, tend to be non-zero or zero together.
	
	The problem \eqref{eq:whole_optimisation} is more difficult to solve than the common $L_2$-minimisation problem (which corresponds to $ \mathbf{G}_{g,t} = \mathbf{I}$, where $\mathbf{I}$ is an identity matrix) or the squared $L_2$-minimisation problem (the problem with $\left \| \mathbf{G}_{g,t} (\cdot) \right \|_2^2$), since the penalty term $\left \| \mathbf{G}_{g,t} (\cdot) \right \|_2$ is nonsmooth. Furthermore, $\mathbf{G}_{g,t}$ is possibly rank-deficient matrix. In this paper, we first derive batch tracking and estimation methods, which are based on the batch computation of the state sequence. To speed up the batch methods, we then propose augmented recursive smoother methods for the primal variable update.

	\section{Batch Tracking and Estimation Methods}
	\label{sec:batch}
	In this section, we introduce the multi-block ADMM framework. Based on this framework, we derive batch algorithms for solving the regularised tracking and state estimation problem.
	
	\subsection{The General Multi-block ADMM (mADMM) Framework}
	\label{sec:admm_framework}
	
	The methods that we develop are based on the multi-block ADMM~\cite{Boyd2011admm}. The multi-block ADMM provides an algorithmic framework which is applicable to problems of the form \eqref{eq:whole_optimisation}, and it can be instantiated by defining the auxiliary variables and their update steps. We introduce auxiliary variables $\mathbf{v}_{t}$ and $\mathbf{w}_{g,t}$, $g=1,\ldots,N_g$, $t=1,\ldots,T$, and then build the following constraints
	\begin{equation}\label{eq:constraint_condition} 
		\begin{aligned}
			\mathbf{x}_t - \mathbf{B}_t \, \mathbf{x}_{t-1} - \mathbf{d}_t &= \mathbf{v}_t, \\
			\mathbf{w}_{1,t} &= \mathbf{G}_{1,t} \,\mathbf{v}_t, \\
			&\,\,\, \vdots   \\
			\mathbf{w}_{N_g,t} &= \mathbf{G}_{N_g,t} \, \mathbf{v}_t.
		\end{aligned}
	\end{equation}
	Note that in \eqref{eq:constraint_condition} we could alternatively introduce auxiliary variables ${\mathbf{w}}_{g,t} = \mathbf{G}_{g,t}(\mathbf{x}_t - \mathbf{B}_t \, \mathbf{x}_{t-1} - \mathbf{d}_t)$, but this replacement would require $\mathbf{G}_{g,t}$ to be invertible when using the augmented recursive smoothers later on. To avoid such restrictions, we employ variables $\mathbf{v}_{t}$ and $\mathbf{w}_{g,t}$ to build the more general constraints in this paper.

	For simplicity of notation, we denote
	${\mathbf{w}}_t = \begin{bmatrix} \mathbf{w}_{1,t}^\top, & \ldots, & \mathbf{w}_{N_g,t}^\top  \end{bmatrix}^\top$, ${\mathbf{G}}_t = \begin{bmatrix} \mathbf{G}_{1,t}^\top, & \ldots, & \mathbf{G}_{N_g,t}^\top\end{bmatrix}^\top$,
	and then solve \eqref{eq:whole_optimisation}, using an equivalent constrained optimisation problem
	\begin{equation}\label{eq:linear_optimisation} 
		\begin{split}
			\begin{aligned}
				& \min_{\substack{\mathbf{x}_{1:{T}},\mathbf{w}_{1:{T}},\\ \mathbf{v}_{1:T}}}
				\frac{1}{2} \sum_{t=1}^{T}  \left\| \mathbf{y}_t - \mathbf{h}_t(\mathbf{x}_t) \right\|_{\mathbf{R}_t^{-1}}^2      +  \sum_{t=1}^{T} \sum_{g = 1}^{N_g} \mu \left \| \mathbf{w}_{g,t}\right \|_2   \\
				& \quad + \frac{1}{2} \sum_{t=2}^{T}  \left \| \mathbf{x}_t - \mathbf{a}_t(\mathbf{x}_{t-1}) \right\|_{\mathbf{Q}_t^{-1}}^2  
				+\frac{1}{2} \left\| \mathbf{x}_1  -   \mathbf{m}_1 \right\|_{ \mathbf{P}_1^{-1} }^2 \\
				&\quad {\mathrm{s.t.}}\,  
				\begin{bmatrix}
					\mathbf{x}_t - \mathbf{B}_t \, \mathbf{x}_{t-1} - \mathbf{d}_t \\
					{\mathbf{w}}_t \end{bmatrix}
				= \begin{bmatrix}
					\mathbf{I} \\
					{\mathbf{G}}_t \end{bmatrix} 
				\mathbf{v}_{t}, \quad  t = 1,\ldots,T. 
			\end{aligned}
		\end{split}
	\end{equation}
	The variables $ \mathbf{x}_{1:T}$, $\mathbf{w}_{1:T}$, and $\mathbf{v}_{1:T}$ can be handled by defining the augmented Lagrangian function
	\begin{equation}\label{eq:Lagrangian_general}
		\begin{split}
			\begin{aligned}
				&\mathcal{L}_{\gamma}(\mathbf{x}_{1:T},{\mathbf{w}_{1:T}},\mathbf{v}_{1:T};\bm{\eta}_{1:T}) \triangleq  
				\frac{1}{2} \sum_{t=1}^{T}  \left\| \mathbf{y}_t - \mathbf{h}_t(\mathbf{x}_t) \right\|_{\mathbf{R}_t^{-1}}^2\\
				& + \frac{1}{2} \sum_{t=2}^{T}  \left\| \mathbf{x}_t - \mathbf{a}_t(\mathbf{x}_{t-1}) \right\|_{\mathbf{Q}_t^{-1}}^2  
				+\frac{1}{2} \left\| \mathbf{x}_1  -   \mathbf{m}_1 \right\|_{ \mathbf{P}_1^{-1} }^2 \\
				&\, + 
				\sum_{t = 1}^{T} \sum_{g = 1}^{N_g}\mu \left \| \mathbf{w}_{g,t} \right \|_2 
				+  \sum_{t = 1}^{T} \bm{\eta}_t^\T \left(\begin{bmatrix}
					\mathbf{u}_t\\
					{\mathbf{w}}_{t} \end{bmatrix}
				- \begin{bmatrix}
					\mathbf{I} \\
					{\mathbf{G}}_{t}  \end{bmatrix}
				\mathbf{v}_t \right)  \\
				&\qquad+  \sum_{t = 1}^{T}  \frac{\gamma}{2} \left\| \begin{bmatrix}
					\mathbf{u}_t \\
					{\mathbf{w}}_{t} \end{bmatrix}
				- \begin{bmatrix}
					\mathbf{I} \\
					{\mathbf{G}}_{t}\end{bmatrix}
				\mathbf{v}_t \right\|_2^2,
			\end{aligned}
		\end{split}
	\end{equation}
	where $\mathbf{u}_t = \mathbf{x}_t - \mathbf{B}_t \, \mathbf{x}_{t-1} - \mathbf{d}_t$, $\bm{\eta}_{t} \in \mathbb{R}^{(N_x  + P_g \times N_g ) } $ is a Lagrangian multiplier, and $\gamma > 0 $ is a penalty parameter. 
	
	The multi-block ADMM (mADMM) framework minimises the function $\mathcal{L}_{\gamma}$ by alternating the $\mathbf{x}_{1:T}$-minimisation step, the $\mathbf{w}_{1:T}$-minimisation step, the $\mathbf{v}_{1:T}$-minimisation step, and the dual variable $\bm{\eta}_{1:T}$ update step. Given $(\mathbf{x}_{1:T}^{(k)},\mathbf{w}_{1:T}^{(k)},\mathbf{v}_{1:T}^{(k)},\bm{\eta}_{1:T}^{(k)})$, the iteration of mADMM has the following steps:
	\begin{subequations}\label{eq:subproblem_dynamic} 
		\begin{align}
			\label{eq:x-primal}
			&\mathbf{x}_{1:T}^{(k+1)}
			= \arg\min_{\mathbf{x}_{1:T}}
			\sum_{t=1}^{T} \frac{1}{2} \left \| {\mathbf{y}_t}- \mathbf{h}_t(\mathbf{x}_t)\right \|_{\mathbf{R}_t^{-1}}^2 + \frac{1}{2} \left \|  \mathbf{x}_1 - \mathbf{m}_1   \right \|_{\mathbf{P}_1^{-1}}^2\notag \\ 
			& + \frac{1}{2} \sum_{t=2}^{T} \|\mathbf{x}_{t} - \mathbf{a}_t(\mathbf{x}_{t-1}) \|_{\mathbf{Q}_t^{-1}}^2 
			+ \frac{\gamma}{2}  \sum_{t=1}^T \left\|\mathbf{u}_t -\mathbf{v}_t^{(k)}  +  \frac{\overline{\bm{\eta}}_t^{(k)}}{\gamma} \right \|_2^2, \\
			\label{eq:w-primal}
			& {\mathbf{w}}_{t}^{(k+1)}  = \arg\min_{\mathbf{w}_{t}}
			\sum_{g = 1}^{N_g}\mu  \left \| \mathbf{w}_{g,t} \right \|_2  
			+  \frac{\gamma}{2} \left\|{\mathbf{w}}_{t} -{\mathbf{G}}_{t} \mathbf{v}_t^{(k)} + \frac{\underline{\bm{\eta}}_{t}^{(k)}}{\gamma} \right\|_2^2     \\
			\label{eq:v-primal}
			&\mathbf{v}_{t}^{(k+1)} = \arg\min_{\mathbf{v}_t} \frac{\gamma}{2} \left\| \begin{bmatrix}
				\mathbf{u}_t^{(k+1)} \\
				{\mathbf{w}}_{t}^{(k+1)} \end{bmatrix}
			- \begin{bmatrix}
				\mathbf{I} \\
				{\mathbf{G}}_{t} \end{bmatrix}
			\mathbf{v}_t  + \frac{\bm{\eta}_t^{(k)}}{\gamma} \right\|_2^2,  \\
			\label{eq:eta-primal}
			&\bm{\eta}_t^{(k+1)} = \bm{\eta}_t^{(k)} + \gamma \left(\begin{bmatrix}
				\mathbf{u}_t^{(k+1)}   \\
				{\mathbf{w}}_{t}^{(k+1)}\end{bmatrix}
			- \begin{bmatrix}
				\mathbf{I} \\
				{\mathbf{G}}_{t}\end{bmatrix}
			\mathbf{v}_t^{(k+1)} \right),
		\end{align}
	\end{subequations}
	where ${\bm{\eta}_t} =\operatorname{vec}(\overline{\bm{\eta}}_t,\underline{\bm{\eta}}_{1,t},\ldots,\underline{\bm{\eta}}_{N_g,t})$. We solve the $\mathbf{w}_{t}$, $\mathbf{v}_{t}$, and $\bm{\eta}_{t}$ subproblems for each $t$, respectively. The $\mathbf{w}_{t}$-subproblem and $\mathbf{v}_{t}$-subproblem have the solutions
	\begin{subequations}\label{eq:w-primal-t} 
		\begin{align}
			\label{eq:wt_update}
			&	\mathbf{w}_{t}^{(k+1)}  = \mathcal{S}_{\mu/\gamma} \left( \mathbf{G}_{g,t} \mathbf{v}_t^{(k)} - \underline{\bm{\eta}}_{g,t}^{(k)}/\gamma\right), \\
			\label{eq:vt_update}
			&\mathbf{v}_t^{(k+1)}   =  \frac{1}{\gamma}( \mathbf{I} + {\mathbf{G}}_t^\T {\mathbf{G}}_t)^{-1}  \left(  
			\begin{bmatrix}
				\mathbf{I} \\
				{\mathbf{G}}_t  \end{bmatrix}^\T  
			\left( \gamma  \begin{bmatrix}
				\mathbf{u}_t^{(k+1)}   \\
				{\mathbf{w}}_t^{(k+1)} \end{bmatrix} + \bm{\eta}_t^{(k)}\right)
			\right),
		\end{align}
	\end{subequations}
	where $\mathcal{S}_{\mu/\gamma} (\cdot) $ is the shrinkage operator \cite{Puig2011shrinkage}. 
	
	Given the mADMM framework, the solutions in \eqref{eq:wt_update}, \eqref{eq:vt_update}, and \eqref{eq:eta-primal} are the basic steps of our methods. In a single iteration, the $\mathbf{w}_{t}$-update can be computed in $\mathcal{O}(N_g)$ operations, and each $\mathbf{v}_{t}$-update takes $\mathcal{O}(N_x^3)$. However, when the $\mathbf{x}_{1:T}$-subproblem is solved by the batch estimation methods, it typically takes $\mathcal{O}(N_x^3 T^3)$ operations. Thus the main computational demand is in updating $\mathbf{x}_{1:T}$. Our main goal here is to derive efficient methods for the $\mathbf{x}_{1:T}$-minimisation step. Before that, we first develop batch methods to solve the $\mathbf{x}_{1:T}$-subproblem.

	\subsection{Batch Solution for Affine Systems}
	\label{sec:batch_linear}
	
	The first batch method we explore is for affine Gaussian systems. We first stack all the state variables into single variables, and then rewrite the $\mathbf{x}_{1:T}$-subproblem \eqref{eq:x-primal} in the form
	\begin{equation}\label{eq:x_primal_batch} 
		\begin{split}
			\begin{aligned}
				&\mathbf{x}^\star = \arg\min_{\mathbf{x}}
				\frac{1}{2} \left \| {\mathbf{y}}- \mathbf{H}\,\mathbf{x} - \mathbf{e}\right \|_{\mathbf{R}^{-1}}^2 
				+\frac{1}{2} \left \| {\mathbf{m}}- \mathbf{A} \mathbf{x} -  \mathbf{b}\right \|_{\mathbf{Q}^{-1}}^2\\ 
				&\quad  + \frac{\gamma}{2} \left\| \mathbf{\Phi} \mathbf{x}  - \mathbf{d} -\mathbf{v}^{(k)}  +  {\overline{\bm{\eta}}^{(k)} }/{\gamma} \right\|_2^2,
			\end{aligned}
		\end{split}
	\end{equation}
	where we have set
	\begin{equation}
		\label{eq:linear_vector_sets}
		\begin{aligned}
			\mathbf{x} = \mathbf{x}_{1:T}, \quad
			\mathbf{\Phi} = \begin{pmatrix}
				\mathbf{I}   & \mathbf{0} & & \\
				-\mathbf{B}_2  &     \mathbf{I}   &  \ddots & \\
				& \ddots & \ddots & \mathbf{0} \\
				& &  -\mathbf{B}_T &  \mathbf{I}  \\
			\end{pmatrix}.
		\end{aligned}
	\end{equation}
	The other variables $\mathbf{y}$, $\mathbf{e}$, $\mathbf{m}$, $\mathbf{d}$, $\mathbf{e}$, $\mathbf{v}$, $\overline{\bm{\eta}}$, $\mathbf{H}$, $\mathbf{R}$, $\mathbf{Q}$, $\mathbf{A}$ are defined analogously to Equation (17) in \cite{Gao2019ieks}.  
	By setting the derivative to zero, the solution is
	\begin{equation}
		\label{eq:x_linear_solution}
		\begin{split}
			& \mathbf{x}^{(k+1)}  = \left( \mathbf{H}^\T \mathbf{R}^{-1} \mathbf{H} + \mathbf{A}^\T \mathbf{Q}^{-1} \mathbf{A}
			+ \gamma \mathbf{\Phi}^\T \mathbf{\Phi} \right)^{-1} \\
			&\quad  \times \big( \mathbf{H}^\T \mathbf{R}^{-1} ( \mathbf{y} -  \mathbf{e}) + 
			\mathbf{A}^\T \mathbf{Q}^{-1} ( \mathbf{m} -  \mathbf{b})   \\
			&\qquad +\gamma\mathbf{\Phi}^\T (\mathbf{d} + \mathbf{v}^{(k)}  -  \overline{\bm{\eta}}^{(k)}/\gamma)  \big).
		\end{split}
	\end{equation}
	In other words, computing the $\mathbf{x}$-minimisation amounts to solving a linear system with the coefficient matrix $\mathbf{H}^\T \mathbf{R}^{-1} \mathbf{H} + \mathbf{A}^\T \mathbf{Q}^{-1} \mathbf{A} + \gamma \mathbf{\Phi}^\T \mathbf{\Phi}$. When the matrix inverse exists, the $\mathbf{x}$-subproblem has a unique solution. Additionally, with a sparsity assumption on the states $\mathbf{x}_{t}$, $\mathbf{\Phi}$ is an identity matrix, and $\mathbf{d}$ is a zero vector. When the  noise $\mathbf{q}_{t}$ is sparse, we can set
	\begin{equation} 
		\begin{aligned}
			\mathbf{\Phi} = \mathbf{A}, \quad 
			\mathbf{d} = \mathbf{m} -  \mathbf{b},
		\end{aligned}
	\end{equation}
	which corresponds to the setting of $\mathbf{B}_t$ and $\mathbf{d}_t$ according to Table~\ref{tab:selection_sparsity}.
	
	The disadvantage of the batch solution is that it requires an extensive amount of computations when $T$ is large. For this reason, in Section \ref{sec:linear_proposed}, we propose to use an augmented recursive smoother, which is mathematically equivalent to the batch method, to improve the computational performance. 
	
	\subsection{Gauss--Newton (GN) for Nonlinear Systems}
	\label{sec:batch_nonlinear}
	
	When the system is nonlinear, we use a similar batch notation as in the affine case, and additionally define the nonlinear functions
	\begin{equation}\label{sets}
		\begin{aligned}
			\mathbf{a}(\mathbf{x}) &=  \operatorname{vec}(\mathbf{x}_1, \mathbf{x}_2 - \mathbf{a}_2(\mathbf{x}_1),\dots,\mathbf{x}_{T} -\mathbf{a}_{T}(\mathbf{x}_{T-1})),\\
			\mathbf{h}(\mathbf{x}) &= \operatorname{vec}(\mathbf{h}_1(\mathbf{x}_1),\dots,\mathbf{h}_{T}(\mathbf{x}_{T})).
		\end{aligned}
	\end{equation}
	The primal $\mathbf{x}_{1:T}$-subproblem then has the form
	\begin{equation}
		\label{eq:x_nonlinear_function}
		\begin{split}
			\mathbf{x}^{(k+1)}  = \arg\min_{\mathbf{x}} \theta(\mathbf{x}),
		\end{split}
	\end{equation}
	where
	\begin{equation}
		\begin{aligned} \label{eq:gn_function}
			&\theta(\mathbf{x}) = \frac{1}{2} \left \| {\mathbf{y}}- \mathbf{h}(\mathbf{x})\right \|_{\mathbf{R}^{-1}}^2 
			+\frac{1}{2} \left \|  \mathbf{m} - 	\mathbf{a}(\mathbf{x}) \right \|_{\mathbf{Q}^{-1}}^2 \\
			&\quad + \frac{\gamma}{2} \left\| \mathbf{\Phi} \mathbf{x}  - \mathbf{d} -\mathbf{v}^{(k)}  +  \overline{\bm{\eta}}^{(k)} /\gamma \right\|_2^2.
		\end{aligned}
	\end{equation}
	The function $\theta(\mathbf{x})$ can now be minimised by the Gauss--Newton (GN) method \cite{Wright2006Numerical}. 
	In GN, we first linearise the nonlinear functions $\mathbf{a}(\mathbf{x}) $ and $\mathbf{h}(\mathbf{x})$,
	and then replace them in $\theta(\mathbf{x})$ by the linear (or actually affine) approximations. The GN iteration then becomes
	\begin{equation}
		\label{eq:x_GN_iteration}
		\begin{split}
			& \mathbf{x}^{(k,i+1)} = \left( \mathbf{J}_{\theta}^\T \mathbf{J}_{\theta}(\mathbf{x}^{(k,i)}) \right)^{-1} 
			\bigg[  \mathbf{J}_h^\T (\mathbf{x}^{(k,i)}) \mathbf{R}^{-1}
			\bigg( \mathbf{y}  - \mathbf{h}(\mathbf{x}^{(k,i)}) \\
			&+ \mathbf{J}_{h}(\mathbf{x}^{(k,i)}) \mathbf{x}^{(k,i)} \bigg)  
			+ \mathbf{J}_a^\T(\mathbf{x}^{(k,i)}) \mathbf{Q}^{-1} \bigg( \mathbf{m}  - \mathbf{a}(\mathbf{x}^{(k,i)})  \\  
			&\, + \mathbf{J}_{a}(\mathbf{x}^{(k,i)}) \,\mathbf{x}^{(k,i)} \bigg)  
			+  \gamma \mathbf{\Phi}^\T \left(   \mathbf{d}+ \mathbf{v}^{(k)} - \overline{\bm{\eta}}^{(k)}/\gamma \right)  \bigg],
		\end{split}
	\end{equation}
	where  
	$$\mathbf{J}_{\theta}^\T \mathbf{J}_{\theta}(\mathbf{x}) 
	= \mathbf{J}_{h}^\T(\mathbf{x}) \mathbf{R}^{-1} \mathbf{J}_{h}(\mathbf{x})	
	+ \mathbf{J}_{a}^\T(\mathbf{x})  \mathbf{Q}^{-1} \mathbf{J}_{a}(\mathbf{x}) +  \gamma \mathbf{\Phi}^\T \mathbf{\Phi}.$$
	
	The above computations are carried out iteratively until a maximum number of iterations $I_{\max}$ is reached. We take the solution $\mathbf{x}^{(k,I_{\max})}$ as the next iterate $\mathbf{x}^{(k+1)}$. 
	While GN avoids the trouble of computing the Hessians of the model functions, it has problems when the Jacobians are rank-deficient. The Levenberg--Marquardt method is introduced next to address this problem.

	\subsection{Levenberg--Marquardt (LM) Method}
	\label{sec:nonlinear_LM}
	The Levenberg--Marquardt  (LM) method \cite{LM1978}, also called as the regularised or damped GN method, improves the performance of GN by using an additional regularisation term.
	With damping factors $\lambda^{(i)} >0$ and a sequence of positive definite regularisation matrices $\mathbf{S}^{(i)}$, the function $\theta(\mathbf{x})$ can be approximated by
	\begin{equation}
		\begin{split} \label{eq:function_app_LM}
			&{\theta} (\mathbf{x}) \approx 
			\frac{1}{2} \left \| {\mathbf{y}}- \mathbf{h}(\mathbf{x}^{(i)}) + \mathbf{J}_{h}(\mathbf{x}^{(i)}) (\mathbf{x} - \mathbf{x}^{(i)}))\right \|_{\mathbf{R}^{-1}}^2 \\ 
			& +\frac{1}{2} \left \|  \mathbf{m} - \mathbf{a}(\mathbf{x}^{(i)}) + \mathbf{J}_{a}(\mathbf{x}^{(i)}) (\mathbf{x} - \mathbf{x}^{(i)}) \right \|_{\mathbf{Q}^{-1}}^2   \\
			& + \frac{\gamma}{2} \left\| \mathbf{\Phi} \mathbf{x}  - \mathbf{d} -\mathbf{v}^{(k)}  +  {\overline{\bm{\eta}}^{(k)} }/{\gamma} \right\|_2^2 
			+ \frac{\lambda^{(i)}}{2} \left\|  \mathbf{x} -  \mathbf{x}^{(i)} \right\|_{[{\mathbf{S}^{(i)}]}^{-1}}^2.
		\end{split}
	\end{equation}
	Using the minimum of this approximate cost function at each step $i$ as the next iterate, we get the following iteration:
	\begin{equation}
		\label{eq:x_LM_iteration}
		\begin{split}
			&\mathbf{x}^{(k,i+1)} = \left(\mathbf{J}_{\theta}^\T \mathbf{J}_{\theta}(\mathbf{x}^{(k,i)} )  + \lambda^{(i)} {[{\mathbf{S}^{(i)}]}^{-1}}\right)^{-1}   \\
			& \bigg[  \mathbf{J}_h^\T(\mathbf{x}^{(k,i)}) \mathbf{R}^{-1}
			\bigg( \mathbf{y}  - \mathbf{h}(\mathbf{x}^{(k,i)}) 
			+ \mathbf{J}_{h}(\mathbf{x}^{(k,i)}) \mathbf{x}^{(k,i)} \bigg)   \\
			&+ \mathbf{J}_a^\T(\mathbf{x}^{(k,i)}) \mathbf{Q}^{-1} \bigg( \mathbf{m}  - \mathbf{a}(\mathbf{x}^{(k,i)})   
			+ \mathbf{J}_{a}(\mathbf{x}^{(k,i)}) \,\mathbf{x}^{(k,i)} \bigg)  \\
			&\quad +  \gamma \mathbf{\Phi}^\T \left(   \mathbf{d}+ \mathbf{v}^{(k)} - \overline{\bm{\eta}}^{(k)}/\gamma \right)  \bigg],
		\end{split}
	\end{equation}
	which is the LM method, when augmented with an adaptation scheme for the regularisation parameters $\lambda^{(i)} > 0$. The regularisation parameter here helps to overcome some problematic cases, for example, the case when $\mathbf{J}_{\theta}^\T \mathbf{J}_{\theta}(\mathbf{x}) $ is rank-deficient, by ensuring the existence of the unique minimum of the approximate cost function. 
	
	At each mADMM iteration, the computation in the $\mathbf{x}_{1:T}$-subproblem such as \eqref{eq:x_linear_solution}, \eqref{eq:x_GN_iteration}, and \eqref{eq:x_LM_iteration}, has a high cost when $T$ is large (e.g., $T = 10^8$). As discussed above, when the main computational demand is indeed in the update of $\mathbf{x}_{1:T}$. Therefore, we utilise the equivalence between batch solutions and recursive smoothers, and then develop efficient augmented recursive smoother methods for solving the $\mathbf{x}_{1:T}$-subproblem.

	\section{Augmented Recursive Smoothers}
	\label{sec:proposed}
	In the section, we will present the augmented KS, GN-IEKS, and LM-IEKS methods for solving the $\mathbf{x}_{1:T}$-subproblem.
	
	\subsection{Augmented Kalman Smoother (KS) for Affine Systems}
	\label{sec:linear_proposed}
	Solving the $\mathbf{x}_{1:T}$-subproblem involves minimisation of a quadratic optimisation problem, which can be efficiently solved by Kalman smoother (KS), see \cite{Gao2020Thesis} for details. We rewrite the batch minimisation problem \eqref{eq:x_primal_batch} as
	\begin{equation}\label{eq:MAP_primal_state}
		\begin{split}
			\begin{aligned}
				&\mathbf{x}_{1:T}^\star
				= \arg\min_{\mathbf{x}_{1:T}} \frac{1}{2} \sum_{t=1}^T\left \| {\mathbf{y}}_t- \mathbf{H}_t \mathbf{x}_t - \mathbf{e}_t \right \|_{\mathbf{R}_t^{-1}}^2 \\
				&  +\frac{1}{2} \sum_{t=2}^T \left \|  \mathbf{x}_t - 	\mathbf{A}_t \mathbf{x}_{t-1} - \mathbf{b}_t   \right \|_{\mathbf{Q}_t^{-1}}^2 
				+\frac{1}{2} \left \|\mathbf{x}_1 - \mathbf{m}_1  \right \|^2_{\mathbf{P}_1^{-1}}\\
				&\quad + \frac{\gamma}{2}  \sum_{t=2}^T \left\|  \mathbf{x}_t - 	\mathbf{B}_t \mathbf{x}_{t-1}   - \mathbf{d}_t -\mathbf{v}_t  +  \frac{\overline{\bm{\eta}}_t}{\gamma} \right\|_2^2  \\
				&\qquad + \frac{\gamma}{2} \left\| \mathbf{x}_1  - \mathbf{m}_1 -\mathbf{v}_1  +  \frac{\overline{\bm{\eta}}_1}{\gamma}  \right \|_2^2.
			\end{aligned}
		\end{split}
	\end{equation}
	It is worth noting that when $\mathbf{B}_t = \mathbf{0}$ and $\mathbf{d}_t =  \mathbf{0}$, the cost function corresponds to the function minimised by KS, which leads to a similar method as was presented in \cite{Gao2019ieks}. For notational convenience, we leave out the iteration number $k$ of mADMM in the following.
	
	Here we consider the general case where $\mathbf{B}_t $ and $\mathbf{d}_t$ are non-zero. Such case is more complicated as we cannot have two dynamic models in a state-space model. For building a dynamic state-space model, we need to fuse the terms in the pairs $ \frac{1}{2}\|  \mathbf{x}_t -	\mathbf{A}_t \mathbf{x}_{t-1} - \mathbf{b}_t    \|_{\mathbf{Q}_t^{-1}}^2$ and $\frac{1}{2} \|  \mathbf{x}_t - 	\mathbf{B}_t \mathbf{x}_{t-1}   - \mathbf{d}_t -\mathbf{v}_t  +  \overline{\bm{\eta}}_t /\gamma \|_2^2 $, along with $\frac{1}{2}\|\mathbf{x}_1 - \mathbf{m}_1  \|^2_{\mathbf{P}_1^{-1}}$ and $\frac{1}{2} \| \mathbf{x}_1  - \mathbf{m}_1 -\mathbf{v}_1  +  \overline{\bm{\eta}}_1/\gamma  \|_2^2$ into single terms. We combine matrices $\mathbf{A}_t$ and $\mathbf{B}_t$ to an artificial transition matrix $\tilde{\mathbf{A}}_t$, fuse $\mathbf{b}_t$ and $(\mathbf{d}_t + \mathbf{v}_t - \overline{\bm{\eta}}_t /\gamma )$ to an artificial bias $\tilde{\mathbf{b}}_t$, and introduce an artificial covariance $\tilde{\mathbf{Q}}_t$, which yields
	\begin{equation}\label{eq:new_system_setting}
		\begin{aligned}
			\tilde{\mathbf{A}}_t  &= (\mathbf{Q}_t^{-1} + \gamma \mathbf{I})^{-1}
			( \mathbf{Q}_t^{-1} \mathbf{A}_t + \gamma  \mathbf{B}_t), \\
			\tilde{\mathbf{b}}_t &=  (\mathbf{Q}_t^{-1} + \gamma \mathbf{I})^{-1}
			(\mathbf{Q}_t^{-1}\mathbf{b}_t + \gamma \mathbf{d}_t  +   \gamma \mathbf{v}_t -   \overline{\bm{\eta}}_t), \\
			\tilde{\mathbf{Q}}_t^{-1} &= \mathbf{Q}_t^{-1} + \gamma \mathbf{I}.
		\end{aligned}
	\end{equation}

	Now, the new artificial dynamic model \eqref{eq:new_system_setting} allows us to use KS to solve the minimisation problem. The problem \eqref{eq:MAP_primal_state} becomes 
	\begin{equation}\label{eq:madmm_xt_double_re_fun} 
		\begin{split}
			\begin{aligned}
				&\mathbf{x}_{1:T}^\star
				= \arg\min_{\mathbf{x}_{1:T}} \frac{1}{2} \sum_{t=1}^T\left \| {\mathbf{y}}_t- \mathbf{H}_t \mathbf{x}_t - \mathbf{e}_t \right \|_{\mathbf{R}_t^{-1}}^2 \\
				&\,+ \frac{1}{2}  \| \mathbf{x}_t -	\tilde{\mathbf{A}}_t \mathbf{x}_{t-1} - \tilde{\mathbf{b}}_t   \|_{\tilde{\mathbf{Q}}_t^{-1}}^2  
				+ \frac{1}{2} \| \mathbf{x}_1  - \tilde{\mathbf{m}}_1 \|^2_{\tilde{\mathbf{P}}_1^{-1}},  \\
			\end{aligned}
		\end{split}
	\end{equation}
	which corresponds to a state-space model, where additionally the initial state has mean $\tilde{\mathbf{m}}_1 =  (\mathbf{P}_1^{-1} + \gamma \mathbf{I})^{-1}$
	$(\mathbf{P}_1^{-1}\mathbf{m}_1 +   \gamma \mathbf{m}_1  +   \gamma \mathbf{v}_t -   \overline{\bm{\eta}}_t)$ and covariance 
	$\tilde{\mathbf{P}}_1^{-1}  =\mathbf{P}_1^{-1} + \gamma \mathbf{I}$. The solution in \eqref{eq:madmm_xt_double_re_fun} can be then computed by running KS on the augmented state-space model
	\begin{subequations}\label{eq:kf_pdf_gauss}
		\begin{align}
			p(\mathbf{x}_t \mid  \mathbf{x}_{t-1}) &=  \mathcal{N} (\mathbf{x}_t  \mid \tilde{\mathbf{A}}_t \mathbf{x}_{t-1} + \tilde{\mathbf{b}}_t  ,\tilde{\mathbf{Q}}_t), \\
			p(\mathbf{y}_t \mid  \mathbf{x}_t) &= \mathcal{N} (\mathbf{y}_t  \mid  \mathbf{H}_t \mathbf{x}_{t} + \mathbf{e}_{t}, \mathbf{R}_t).
		\end{align}
	\end{subequations}
	The augmented KS requires only $\mathcal{O}(N_x^3 T)$ operations which is much less than the corresponding batch solution in \eqref{eq:x_linear_solution}. The augmented KS method is summarised in Algorithm~\ref{alg:augmented_KS}.

	\subsection{Gauss--Newton IEKS (GN-IEKS) for Nonlinear Systems}
	\label{sec:nonlinear_ieks}
	
	The solution of \eqref{eq:gn_function} has similar computational scaling challenges as the affine case discussed in previous section. However, we can use the equivalence of IEKS and GN \cite{Bell1994smoother} to construct an efficient iterative solution for the optimisation problem in the primal space. In the GN-IEKS method, we first approximate the nonlinear model by linearisation, and then use KS on the linearised model. The $\mathbf{x}_{1:T}$-subproblem now takes the form of \eqref{eq:x-primal}. 
	In IEKS, at $i$:th iteration, we form affine approximations of $\mathbf{a}_t(\mathbf{x}_{t-1})$ and $\mathbf{h}_t(\mathbf{x}_t)$ as follows:
	\begin{equation}
		\label{eq:ieks_nonliear_appro}
		\begin{aligned}
			\mathbf{a}_t(\mathbf{x}_{t-1}) &\approx \mathbf{a}_t(\mathbf{x}_{t-1}^{(i)}) + \mathbf{J}_{a_t}(\mathbf{x}_{t-1}^{(i)}) (\mathbf{x}_{t-1} - \mathbf{x}_{t-1}^{(i)}), \\
			\mathbf{h}_t(\mathbf{x}_t) &\approx \mathbf{h}_t(\mathbf{x}_t^{(i)}) + \mathbf{J}_{h_t}(\mathbf{x}_t^{(i)}) (\mathbf{x}_t - \mathbf{x}_t^{(i)}).
		\end{aligned}
	\end{equation}
	We replace the nonlinear functions in the cost function with the above approximations, and compute the next iterate as the solution to the minimisation problem
	\begin{equation}\label{eq:ieks_linearisation}
		\begin{split}
			\begin{aligned}
				& \mathbf{x}_{1:T}^{(i+1)}
				= \arg\min_{\mathbf{x}_{1:T}}
				\frac{1}{2} \left \| {\mathbf{y}_t} -\mathbf{h}_t(\mathbf{x}_t^{(i)}) + \mathbf{J}_{h_t}(\mathbf{x}_t^{(i)}) (\mathbf{x}_t- \mathbf{x}_t^{(i)})\right \|_{\mathbf{R}_t^{-1}}^2 \\ 
				&\, + \frac{1}{2} \sum_{t=2}^{T} \left\|\mathbf{x}_{t} - \mathbf{a}_t(\mathbf{x}_{t-1}^{(i)}) + \mathbf{J}_{a_t}(\mathbf{x}_{t-1}^{(i)}) (\mathbf{x}_{t-1} - \mathbf{x}_{t-1}^{(i)}) \right\|_{\mathbf{Q}_t^{-1}}^2 
				\\
				&\quad + \frac{\gamma}{2}  \sum_{t=2}^T \left\|  \mathbf{x}_t - 	\mathbf{B}_t \,\mathbf{x}_{t-1}   - \mathbf{d}_t -\mathbf{v}_t  +  \frac{\overline{\bm{\eta}}_t}{\gamma}  \right\|_2^2  \\
				&\quad \quad  
				+ \frac{\gamma}{2} \left\| \mathbf{x}_1  - \mathbf{m}_1 -\mathbf{v}_1  +  \frac{\overline{\bm{\eta}}_1}{\gamma}\right\|_2^2 + \frac{1}{2} \left \|  \mathbf{x}_1 - 	\mathbf{m}_1   \right \|_{\mathbf{P}_1^{-1}}^2, 
			\end{aligned}
		\end{split}
	\end{equation}
	which is equivalent to \eqref{eq:MAP_primal_state} with
	\begin{equation} \label{eq:ieks_setting}
		\begin{aligned}
			\mathbf{A}_{t} &=  \mathbf{J}_{a_t}(\mathbf{x}_{t-1}^{(i)}), \quad 
			&\mathbf{b}_t& =  \mathbf{a}_t(\mathbf{x}_{t-1}^{(i)}) - \mathbf{J}_{a_t}(\mathbf{x}_{t-1}^{(i)}) \,\mathbf{x}_{t-1}^{(i)},  \\
			\mathbf{H}_{t} &=\mathbf{J}_{h_t}(\mathbf{x}_t^{(i)}), \quad
			&\mathbf{e}_t &=  \mathbf{h}_t(\mathbf{x}_t^{(i)}) - \mathbf{J}_{h_t}(\mathbf{x}_t^{(i)})\, \mathbf{x}_t^{(i)}.
		\end{aligned}
	\end{equation}
	The precise expressions of $\mathbf{B}_t$ and $\mathbf{d}_t$ depend on our choice of sparsity. When $\mathbf{q}_{t}$ is sparse, the expressions are given by 
	\begin{equation} \label{eq:Bt_setting}
		\begin{aligned}
			\mathbf{B}_{t} =  \mathbf{J}_{a_t}(\mathbf{x}_{t-1}^{(i)}), \quad
			\mathbf{d}_t  =  \mathbf{a}_t(\mathbf{x}_{t-1}^{(i)}) - \mathbf{J}_{a_t}(\mathbf{x}_{t-1}^{(i)}) \,\mathbf{x}_{t-1}^{(i)},
		\end{aligned}
	\end{equation}
	which needs the same computations as in \eqref{eq:new_system_setting}. Thus we can solve the minimisation problem in \eqref{eq:x-primal} by iteratively linearising the nonlinearities and then by applying KS. This turns out to be mathematically equivalent to applying GN to the batch problem as we did in Section~\ref{sec:batch_nonlinear}. The steps of the GN-IEKS method are summarised in Algorithm~\ref{alg:GN-IEKS}.
	
	\begin{algorithm}[htb] \label{alg:augmented_KS}
		\caption{Augmented KS} 
		\KwIn{$\mathbf{y}_t$, $\mathbf{B}_t$, $\mathbf{d}_t$, $\mathbf{A}_t$, $\mathbf{H}_t$, $\mathbf{R}_t$, $\mathbf{Q}_t$, $\mathbf{v}^{(k)}$, $\overline{\bm{\eta}}^{(k)}$, $\mathbf{m}_1$, $\mathbf{P}_1$, and $\gamma$.} 
		\KwOut{$\mathbf{x}_{1:T}^{*}$.}     
		compute $\tilde{\mathbf{A}}_t$, $\tilde{\mathbf{Q}}_t$, and $\tilde{\mathbf{b}}_t$ by \eqref{eq:new_system_setting}\;    
		\For{$t = 1,\ldots,T$}
		{	$\mathbf{m}_t^- =  \tilde{\mathbf{A}}_t  \mathbf{m}_{t-1} + \tilde{\mathbf{b}}_t $\;    
			$\mathbf{P}_t^- = \tilde{\mathbf{A}}_t \,\mathbf{P}_{t-1}\,\tilde{\mathbf{A}}_t^\T + \tilde{\mathbf{Q}}_t$\;  
			$\mathbf{S}_t =  \mathbf{H}_t \, \mathbf{P}_t^-  \, \mathbf{H}_t^\T  + \mathbf{R}_t$\;  
			$\mathbf{K}_t = \mathbf{P}_t^- \,  \mathbf{H}_t^\T  \, [\mathbf{S}_t]^{-1}$\;  
			$\mathbf{m}_t = \mathbf{m}_t^-  + \mathbf{K}_t \big(\mathbf{y}_t - ( \mathbf{H}_t \, \mathbf{m}_t^- + \mathbf{e}_t)\big)$\;  
			$	\mathbf{P}_t = \mathbf{P}_t^- - \mathbf{K}_t \, \mathbf{S}_t \,[\mathbf{K}_t]^\T$\;  
		}
		$\mathbf{m}_T^s = \mathbf{m}_T$ \text{and} $\mathbf{P}_T^s = \mathbf{P}_T$\;
		\For{$t = T-1,\ldots,1$}
		{$\mathbf{G}_t   =   \mathbf{P}_t \, \tilde{\mathbf{A}}_{t+1}^\T  \, [\mathbf{P}^-_{t+1}]^{-1}$\;
			$\mathbf{m}_t^s =   \mathbf{m}_t + \mathbf{G}_t  \, \big(\mathbf{m}^s_{t+1} - \mathbf{m}^-_{t+1}\big)$\;
			$\mathbf{P}_t^s =   \mathbf{P}_t + \mathbf{G}_t \, \big(\mathbf{P}^s_{t+1} - \mathbf{P}^-_{t+1}\big) \, \mathbf{G}_t^\T$\;
		}
		return $\mathbf{x}_{1:T}^{*} = \mathbf{m}_{1:T}^s$\; 
	\end{algorithm}
	
	\begin{algorithm}[htb] \label{alg:GN-IEKS}
		\caption{GN-IEKS}  
		\KwIn{$\mathbf{y}_t$, $\mathbf{B}_t$, $\mathbf{d}_t$, $\mathbf{a}_t$, $\mathbf{h}_t$, $\mathbf{R}_t$, $\mathbf{Q}_t$, $\mathbf{v}^{(k)}$, $\overline{\bm{\eta}}^{(k)}$, $\mathbf{m}_1$, $\mathbf{P}_1$, and $\gamma$.} 
		\KwOut{$\mathbf{x}_{1:T}^{*}$.}
		{set $i \leftarrow 0$ and start from a suitable initial guess $\mathbf{x}_{1:T}^{(0)}$}\; 
		\While{not converged or $i < I_{\max}$}
		{linearise $\mathbf{a}_t$ and $\mathbf{h}_t$ according to \eqref{eq:ieks_nonliear_appro}\; 
			compute $\tilde{\mathbf{A}}_t$, $\tilde{\mathbf{Q}}_t$, $\tilde{\mathbf{b}}_t$ by \eqref{eq:new_system_setting}\;
			compute $\mathbf{x}_{1:T}^{(i+1)}$ by \eqref{eq:ieks_linearisation} using the augmented KS\;
			{$i \leftarrow i+1$}\;
		}
		return $\mathbf{x}_{1:T}^{*} = \mathbf{x}_{1:T}^{(i)}$\; 
	\end{algorithm}

	\subsection{Levenberg--Marquardt IEKS (LM-IEKS)}
	\label{sec:nonlinear_LM_IEKS}
	There also exists a connection between the Levenberg--Marquardt (LM) and a modified version of IEKS. The LM-IEKS method \cite{simo2020LMIEKS} is based on replacing the minimisation of the approximate cost function in \eqref{eq:ieks_linearisation} by a regularised minimisation of the form	
	\begin{equation}\label{eq:lm_ieks_linearisation}
		\begin{split}
			\begin{aligned}
				& \mathbf{x}_{1:T}^{\star}
				= \arg\min_{\mathbf{x}_{1:T}}
				\frac{1}{2} \left \| {\mathbf{y}_t}-\mathbf{h}_t(\mathbf{x}_t^{(i)}) + \mathbf{J}_{h_t}(\mathbf{x}_t^{(i)}) (\mathbf{x}_t - \mathbf{x}_t^{(i)})\right \|_{\mathbf{R}_t^{-1}}^2 \\ 
				& + \frac{1}{2} \sum_{t=2}^{T} \left\|\mathbf{x}_{t} - \mathbf{a}_t(\mathbf{x}_{t-1}^{(i)}) + \mathbf{J}_{a_t}(\mathbf{x}_{t-1}^{(i)}) (\mathbf{x}_{t-1} - \mathbf{x}_{t-1}^{(i)}) \right\|_{\mathbf{Q}_t^{-1}}^2 		
				\\
				& + \frac{\gamma}{2}  \sum_{t=2}^T \left\|  \mathbf{x}_t -	\mathbf{B}_t \,\mathbf{x}_{t-1}   - \mathbf{d}_t - \mathbf{v}_t  +  \frac{\overline{\bm{\eta}}_t}{\gamma}  \right\|_2^2 
				+\frac{1}{2} \left \|  \mathbf{x}_1  -	\mathbf{m}_1   \right \|_{\mathbf{P}_1^{-1}}^2  \\
				&\,   
				+ \frac{\lambda^{(i)}}{2} \sum_{t=1}^T \left\|  \mathbf{x}_t -\mathbf{x}_t^{(i)} \right\|_{[\mathbf{S}_t^{(i)}]^{-1}}^2
				+ \frac{\gamma}{2} \left\| \mathbf{x}_1  - \mathbf{m}_1 -\mathbf{v}_1  +  \frac{\overline{\bm{\eta}}_1 }\gamma \right \|_2^2,
			\end{aligned}
		\end{split}
	\end{equation}
	where we have assume that $\mathbf{S}^{(i)} = \operatorname{diag}(\mathbf{S}_1^{(i)}, \ldots, \mathbf{S}_T^{(i)})$. Similarly to GN-IEKS, when $\mathbf{B}_{t}$ and $\mathbf{d}_{t}$ are non-zero, we need to build a new state-space model in order to have only one dynamic model. Following \cite{simo2020LMIEKS}, the regularisation can be implemented by defining an additional pseudo-measurement $\mathbf{z}_t =  \mathbf{x}_t^{(i)}$ with a noise covariance $\mathbf{\Sigma}_t^{(i)} = \mathbf{S}_t^{(i)}/\lambda^{(i)}$. Using \eqref{eq:new_system_setting} and \eqref{eq:ieks_setting} , we have the augmented state-space model 
	\begin{equation}\label{eq:lmeks_pdf_gauss}
		\begin{aligned}
			p(\mathbf{x}_t \mid  \mathbf{x}_{t-1}) &=  \mathcal{N} (\mathbf{x}_t  \mid \tilde{\mathbf{A}}_t \mathbf{x}_{t-1} + \tilde{\mathbf{b}}_t  ,\tilde{\mathbf{Q}}_t), \\
			p(\mathbf{y}_t \mid  \mathbf{x}_t) &= \mathcal{N} (\mathbf{y}_t  \mid  \mathbf{H}_t \mathbf{x}_{t} + \mathbf{e}_{t}, \mathbf{R}_t),\\
			p(\mathbf{z}_t \mid  \mathbf{x}_t) &= \mathcal{N} (\mathbf{z}_t  \mid  \mathbf{x}_{t}, \mathbf{\Sigma}_t^{(i)}),
		\end{aligned}
	\end{equation}
	which provides the minimum of the cost function as the KS solution. By combining this with $\lambda^{(i)}$ adaptation and iterating, we can implement the LM algorithm for the $\mathbf{x}_{1:T}$-subproblem using the recursive smoother (cf.\ \cite{simo2020LMIEKS}).
	See Algorithm~\ref{alg:LM_EKS} for more details.
	
	\begin{algorithm}[htb] 
		\caption{LM-IEKS}   \label{alg:LM_EKS}
		\KwIn{$\mathbf{y}_t$, $\mathbf{B}_t$, $\mathbf{d}_t$, $\mathbf{a}_t$, $\mathbf{h}_t$, $\mathbf{R}_t$, $\mathbf{Q}_t$, $\mathbf{v}^{(k)}$, $\overline{\bm{\eta}}^{(k)}$, $\mathbf{m}_1$ and $\mathbf{P}_1$; $\mathbf{S}_t$, $\gamma$, $\lambda$, and $\alpha$.} 
		\KwOut{$\mathbf{x}_{1:T}^{*}$.}         
		{set $i \leftarrow 0$ and start from a suitable initial guess $\mathbf{x}_{1:T}^{(0)}$}\; 
		\While{not converged or $i < I_{\max}$}
		{linearise $\mathbf{a}_t$ and $\mathbf{h}_t$ according to \eqref{eq:ieks_nonliear_appro} \; 
			compute $\tilde{\mathbf{A}}_t$, $\tilde{\mathbf{Q}}_t$, $\tilde{\mathbf{b}}_t$ by \eqref{eq:new_system_setting}\;
			update $\mathbf{x}_{1:T}^{(i+1)}$ by \eqref{eq:lmeks_pdf_gauss} based on the augmented KS\;
			\eIf{${\theta}(\mathbf{x}_{1:T}^{(i+1)}) < {\theta}(\mathbf{x}_{1:T}^{(i)}) $}
			{$\lambda^{(i)} \leftarrow \lambda^{(i)}/\alpha$; $i \leftarrow i+1$\;}
			{$\lambda^{(i)} \leftarrow \lambda^{(i)} \, \alpha$\;}
		}	
		return $\mathbf{x}_{1:T}^{*} = \mathbf{x}_{1:T}^{(i)}$\;
	\end{algorithm}

	\subsection{Discussion}
	\label{sec:summary}
	
	All the methods discussed above, namely augmented KS, GN-IEKS, and LM-IEKS, provide efficient ways to solve the $\mathbf{x}_{1:T}$-subproblem. When we leverage the Markov structure of the $\mathbf{x}_{1:T}$-subproblem arising in mADMM iteration, we can significantly reduce the computation burden. In particular, when the function $\mathbf{a}_{t}(\mathbf{x}_{t-1})$ and $\mathbf{h}_{t}(\mathbf{x}_{t})$ are affine, the augmented KS method can be used in the $\mathbf{x}_{1:T}$-subproblem (see \eqref{eq:MAP_primal_state}). Both GN-IEKS and LM-IEKS are based on the use of linearisation of the functions $\mathbf{a}_{t}(\mathbf{x}_{t-1})$ and $\mathbf{h}_{t}(\mathbf{x}_{t})$, and they work well for most nonlinear minimisation problems. However, when the Jacobians (e.g., $\mathbf{J}_{a_t}(\mathbf{x}_{t-1}^{(i)})$ or $\mathbf{J}_{h_t}(\mathbf{x}_t^{(i)})$ in \eqref{eq:ieks_nonliear_appro}) are rank-deficient, the GN-IEKS method cannot be used. As a robust extension of GN-IEKS, LM-IEKS significantly improves the performance of GN-IEKS. It should be noted that when the regularisation term is not used in LM-IEKS (when $\lambda^{(i)} = 0$), then LM-IEKS reduces to GN-IEKS \cite{simo2020LMIEKS}.
	
	\section{Convergence Analysis}
	\label{sec:convergence}
	
	In this section, we prove that under mild assumptions and a proper choice of the penalty parameter, our KS-mADMM, GN-IEKS-mADMM, and LM-IEKS-mADMM methods converge to a stationary point of the original problem.
	Although convergence of the multi-block ADMM has already been proven, the existing results strongly depend on convexity assumptions or Lipschitz continuity conditions (see, e.g., \cite{He2012convergence,admm2019convergence,Hong2016Convergence}). In the analysis, we require neither the convexity of the objective function nor Lipschitz continuity conditions. Instead, we use a milder condition on the amenability. This allows us to establish the convergence of the three methods.
	
	For the case when the functions $\mathbf{\mathbf{a}}_t(\mathbf{x}_{t-1})$ and $\mathbf{h}_t(\mathbf{x}_{t})$ are affine (see \eqref{eq:linear_model}), we have the following lemma.
	\begin{lemma}
		\label{lemma:s_sequence}
		Let $ \left\{\mathbf{x}_{1:T}^{(k)},\mathbf{w}_{1:T}^{(k)},\mathbf{v}_{1:T}^{(k)},\bm{\eta}_{1:T}^{(k)}\right\} $ be the iterates generated by \eqref{eq:subproblem_dynamic}. Then we have 
		\begin{equation}\label{eq:convex_H_leq}
			\begin{split}
				\begin{aligned}
					&\left\|\begin{bmatrix}
						\mathbf{v}^{(k+1)}\\
						\bm{\eta}^{(k+1)}\\
					\end{bmatrix}  - \begin{bmatrix}
						\mathbf{v}^\star\\
						\bm{\eta}^\star\\
					\end{bmatrix} \right\|_{\mathbf{\Omega}}^2 \\ 
					&\leq 
					\left\|\begin{bmatrix}
						\mathbf{v}^{(k)}\\
						\bm{\eta}^{(k)}\\
					\end{bmatrix}  - \begin{bmatrix}
						\mathbf{v}^\star\\
						\bm{\eta}^\star\\
					\end{bmatrix} \right\|_{\mathbf{\Omega}}^2     -
					\left\|\begin{bmatrix}
						\mathbf{v}^{(k)}\\
						\bm{\eta}^{(k)}\\
					\end{bmatrix}  - 
					\begin{bmatrix}
						\mathbf{v}^{(k+1)}\\
						\bm{\eta}^{(k+1)} \\
					\end{bmatrix} \right\|_{\mathbf{\Omega}}^2,
				\end{aligned}
			\end{split}
		\end{equation}
		where 
		$\mathbf{\Omega} = \begin{bmatrix}
			\gamma \mathbf{I} + {\mathbf{G}}^\T {\mathbf{G}}  &  \mathbf{0}\\
			\mathbf{0}  &\mathbf{I}/\gamma \\
		\end{bmatrix}$ and 
		${\mathbf{G}} = \begin{bmatrix} \mathbf{G}_{1} \\ \vdots \\ \mathbf{G}_{T} \end{bmatrix}$.
	\end{lemma}
	
	\begin{proof}
		See Appendix \ref{pf:s_sequence}. 
	\end{proof}
	
	We will then establish the convergence rate of the proposed method in terms of the iteration number.
	
	\begin{theorem}[Convergence of KS-mADMM]
		\label{theorem_admm}
		Let $\mathbf{Q}_t$ and $\mathbf{P}_1$ be positive semi-definite matrices. Then the sequence $\{\mathbf{x}_{1:T}^{(k)}, \mathbf{w}_{1:T}^{(k)}, \mathbf{v}_{1:T}^{(k)}, \bm{\eta}_{1:T}^{(k)}\}$ generated by KS-mADMM converges to a stationary point $(\mathbf{x}_{1:T}^\star, \mathbf{w}_{1:T}^\star, \mathbf{v}_{1:T}^\star,\bm{\eta}_{1:T}^\star)$ with the rate ${o}(\frac{1}{k})$. 	
	\end{theorem}
	
	\begin{proof}
		The proof is based on the convexity of the function. Because of the equivalence between mADMM and KS-mADMM, we start by establishing the convergence of mADMM. When $\mathbf{Q}_t$ and $ \mathbf{P}_1$ are positive semi-definite, the function in \eqref{eq:linear_optimisation} is convex. Because of $[\mathbf{\Phi} \quad \mathbf{0}] [\mathbf{0} \quad \mathbf{I}]^\T = \mathbf{0}$, we can write $\mathbf{x}$ and $\mathbf{w}$ into a function $\Xi(\bm{\zeta})$ in the batch form \cite{He2012convergence}. 
		
		For simplicity of notation, we define $\mathbf{s} = \begin{bmatrix} \mathbf{v}  & \bm{\eta} \end{bmatrix}^\top$. Using Lemma \ref{lemma:s_sequence}, we obtain  
		\begin{equation}\label{eq:lemma_step_8}
			\begin{split}
				\begin{aligned}
					&\Xi(\bm{\zeta}^\star) - \Xi(\bm{\zeta}^{(k)})  + (\bm{\xi}^\star - \bm{\xi}^{(k)})^\T   F(\bm{\xi}^\star) \\ 
					& \quad +\| \mathbf{s}^\star -  \mathbf{s}^{(k)} \|_{\mathbf{\Omega}}^2 
					\geq 
					\| \mathbf{s}^\star -  \mathbf{s}^{(k+1)} \|_{\mathbf{\Omega}}^2, 
				\end{aligned}
			\end{split}
		\end{equation}
		where $\bm{\xi}$ and $F(\bm{\xi})$ are defined in Appendix \ref{pf:s_sequence} (see \eqref{eq:subproblem_linear_gradient}).
		We sum the inequality \eqref{eq:lemma_step_8} from $0$ to $k$, and divide each term by $k+1$. Since $\| \mathbf{s}^\star -  \mathbf{s}^{(k+1)} \|_{\mathbf{\Omega}}^2 \geq 0$, we then have
		\begin{equation}\label{eq:lemma_step_9}
			\begin{split}
				\begin{aligned}
					&\frac{1}{k+1} \sum_{i = 0}^{k} \Xi(\bm{\zeta}^{(k)}) -  \Xi(\bm{\zeta}^\star)\\
					&+ \big( \frac{1}{k+1} \sum_{i = 0}^{k}\bm{\xi}^{(k)} - \bm{\xi}^\star \big)^\T   F(\bm{\xi}^\star) 
					\leq \frac{1}{k+1} \| \mathbf{s}^\star -  \mathbf{s}^{(0)} \|_{\mathbf{\Omega}}^2.
				\end{aligned}
			\end{split}
		\end{equation}
		Let $\bar{\bm{\zeta}}^{(k)} = \frac{1}{k+1}\bm{\zeta}^{(k)}$ and $\bar{\bm{\xi}}^{(k)} = \frac{1}{k+1}\bm{\xi}^{(k)}$. 
		Because of the convexity of $\Xi$, we further write \eqref{eq:lemma_step_9} as
		\begin{equation}\label{eq:lemma_step_10}
			\begin{split}
				\begin{aligned}
					\Xi(\bar{\bm{\zeta}}^{(k)})  -  \Xi(\bm{\zeta}^\star)
					+ ( \bar{\bm{\xi}}^{(k)} - \bm{\xi}^\star)^\T   F(\bm{\xi}^\star) 
					\leq \frac{1}{k+1} \| \mathbf{s}^\star -  \mathbf{s}^{(0)} \|_{\mathbf{\Omega}}^2.
				\end{aligned}
			\end{split}
		\end{equation}
		
		The convergence rate ${o}(\frac{1}{k})$ of mADMM is thus established. As batch mADMM is equivalent to KS-mADMM, then the sequence $\{\mathbf{x}^{(k)},\mathbf{w}^{(k)}, \mathbf{v}^{(k)},\bm{\eta}^{(k)}\}$ and the sequence $\{\mathbf{x}_{1:T}^{(k)},\mathbf{w}_{1:T}^{(k)}, \mathbf{v}_{1:T}^{(k)},\bm{\eta}_{1:T}^{(k)}\}$ are identical. This concludes the proof.
	\end{proof}
	
	When the functions $\mathbf{\mathbf{a}}_t(\mathbf{x}_{t-1})$ and $\mathbf{h}_t(\mathbf{x}_{t})$ are nonlinear, we have the function
	\begin{equation}\label{eq:Lagrangian_nonlinear_con}
		\begin{aligned}
			s(\mathbf{x}) \triangleq \frac{1}{2}\| \mathbf{y} -\mathbf{h}(\mathbf{x})  \|_{\mathbf{R}^{-1}}^2 + \frac{1}{2} \left \| \mathbf{m} - \mathbf{a}(\mathbf{x}) \right \|_{\mathbf{Q}^{-1}}^2.
		\end{aligned}
	\end{equation}

	\begin{definition} The function $ s(\mathbf{x}) $ is $\textit{strongly amenable}$ \cite{Rockafellar1998Variational} at $\mathbf{x}$ when the condition
		\begin{equation}\label{eq:amenable_condidtion}
			\begin{split}
				\begin{aligned}
					\begin{bmatrix} 
						\mathbf{R}^{-1/2} \mathbf{J}_{h} ; \mathbf{Q}^{-1/2}\mathbf{J}_{a}\end{bmatrix}^\top \mathbf{z} = \mathbf{0}
				\end{aligned}
			\end{split}
		\end{equation}
		is satisfied only when $\mathbf{z}$ is zero. 
	\end{definition}

	Let $s(\mathbf{x})$ be strongly amenable. Then, $s(\mathbf{x})$ will be \textit{prox-regular} \cite{Poliquin1996Prox} . We are now ready for introducing the following lemma.

	\begin{lemma} [Bounded and nonincreasing sequence]
		\label{lemma:nonincreasing_bounded} Assume that $\delta_{+}(\mathbf{\Phi}^\T \mathbf{\Phi}) > 0$
		and $s(\mathbf{x}) $ is strongly amenable. Then there exists $\gamma > 0$ such that
		sequence $\mathcal{L}_\gamma(\mathbf{x}^{(k)}, \mathbf{w}^{(k)},\mathbf{v}^{(k)}; \bm{\eta}^{(k)})$ is bounded and nonincreasing. 
	\end{lemma}
	
	\begin{proof}
		See Appendix \ref{pf:lemma-nonincreasing}. 
	\end{proof}

	Next, we present the main theoretical result.

	\begin{theorem}[Convergence of GN-IEKS-mADMM]
		\label{theorem:ieks_madmm}
		Let the assumptions in Lemma~\ref{lemma:nonincreasing_bounded} be satisfied.
		Then there exists $\gamma > 0$ such that the sequence $\left\{\mathbf{x}^{(k)}_{1:T}, \mathbf{w}^{(k)}_{1:T},\right.$$\left.\mathbf{v}^{(k)}_{1:T}, \bm{\eta}^{(k)}_{1:T}\right\}$ generated by GN-IEKS-mADMM locally converges to a local minimum.
	\end{theorem}
	\begin{proof}
		By Lemma~\ref{lemma:nonincreasing_bounded}, the sequence $\mathcal{L}_{\gamma}(\mathbf{x}^{(k)}, \mathbf{w}^{(k)}, \mathbf{v}^{(k)}; \bm{\eta}^{(k)})$ is bounded and nonincreasing. 
		Based on our paper \cite{Gao2019ieks}, the $\mathbf{x}$-subproblem has a local minimum $\mathbf{x}^{\star}$. The $\mathbf{w}$ and $\mathbf{v}$ subproblems are convex \cite{boyd2004Convex}. We then conclude that the iterative sequence $\{\mathbf{x}^{(k)}, \mathbf{w}^{(k)}, \mathbf{v}^{(k)}, \bm{\eta}^{(k)}\}$ locally converges to a local minimum $(\mathbf{x}^\star, \mathbf{w}^\star, \mathbf{v}^\star,\bm{\eta}^\star)$. According to\cite{Bell1994smoother}, GN is equivalent to IEKS. Thus we deduce that the iterative sequence $\{\mathbf{x}^{(k)}_{1:T}, \mathbf{w}^{(k)}_{1:T}, \mathbf{v}^{(k)}_{1:T}, \bm{\eta}^{(k)}_{1:T}\}$ is convergent to a local minimum $(\mathbf{x}^\star_{1:T}, \mathbf{w}^\star_{1:T},\mathbf{v}^\star_{1:T}, \bm{\eta}^\star_{1:T})$.
	\end{proof}

	\begin{lemma}[Convergence of LM]
		\label{lemma:LM}
		Assume that the norm of Hessian $\mathbf{H}_{\theta} (\mathbf{x})$ is bounded by a positive constant $\kappa < \max \left\{\gamma \delta_{+}(\mathbf{\Phi}^\top \mathbf{\Phi}),\lambda^{(i)} \delta_{+}([\mathbf{S}^{(i)}]^{-1}) \right\}$. Then LM is locally (linearly) convergent. The convergence is quadratic when $ \kappa \to 0$.
	\end{lemma}
	
	\begin{proof}
		See Appendix \ref{pf:lemma-LM}. 
	\end{proof}
	
	\begin{theorem}[Convergence of LM-IEKS-mADMM]
		Let the assumptions of Lemmas~\ref{lemma:nonincreasing_bounded} and \ref{lemma:LM} be satisfied.
		Then there exists $\lambda^{(i)}, \gamma > 0 $ such that the sequence $\{\mathbf{x}^{(k)}_{1:T}, \mathbf{w}^{(k)}_{1:T},\mathbf{v}^{(k)}_{1:T}, \bm{\eta}^{(k)}_{1:T}\}$ generated by \mbox{LM-IEKS-mADMM} converges to a local minimum $(\mathbf{x}^\star_{1:T}, \mathbf{w}^\star_{1:T}, \mathbf{v}^\star_{1:T},\bm{\eta}^\star_{1:T})$. 
	\end{theorem}
	
	\begin{proof}
		Similarly to Theorem \ref{theorem:ieks_madmm}, we use Lemma~\ref{lemma:nonincreasing_bounded} to establish 
		that the sequence $\mathcal{L}_{\gamma}(\mathbf{x}^{(k)}, \mathbf{w}^{(k)}, \mathbf{v}^{(k)}; \bm{\eta}^{(k)})$ is bounded and nonincreasing. Due to the convexity, the $\mathbf{w}$ and $\mathbf{v}$ subproblems have a local minimum. 
		By Lemma~\ref{lemma:LM}, the sequence ${\mathbf{x}^{(i)}}$ generated by LM converges to $\mathbf{x}^{\star}$. Then the sequence $\{\mathbf{x}^{(k)}, \mathbf{w}^{(k)}, \mathbf{v}^{(k)}, \bm{\eta}^{(k)}\}$ locally converges to a minimum $(\mathbf{x}^\star, \mathbf{w}^\star, \mathbf{v}^\star,\bm{\eta}^\star)$. Since the sequence $\{\mathbf{x}^{(k)}, \mathbf{w}^{(k)}, \mathbf{v}^{(k)}, \bm{\eta}^{(k)}\}$ generated by LM is identical to $\{\mathbf{x}^{(k)}_{1:T}, \mathbf{w}^{(k)}_{1:T}, \mathbf{v}^{(k)}_{1:T}, \bm{\eta}^{(k)}_{1:T}\}$ generated by LM-IEKS \cite{Bell1994smoother,simo2020LMIEKS}. 
	\end{proof}

	\section{Numerical experiments}
	\label{sec:results}
	
	In this section, we experimentally evaluate the proposed methods in a selection of different applications, including linear target tracking problems, multi-sensor range measurement problems, ship trajectory-tracking, audio restoration, and autonomous vehicle tracking. As for the convergence criteria, we can easily verify that the assumptions for convergence are satisfied for the linear/affine examples in Sections \ref{sec:linear_car}, \ref{sec:linear_Marine}, and \ref{sec:audio}. Additionally, the nonlinear coordinated turn model in Section \ref{sec:urban_car} also satisfies assumptions for convergence. However, for the distance measurement in Section \ref{sec:nonlinear_car}, it is hard to establish the strong amenability although empirically the convergence occurs.
	
	\subsection{Linear Target Tracking Problems}
	\label{sec:linear_car}
	
	In the first experiment, we consider simulated tracking of a moving target (such as car) with the Wiener velocity model~\cite{Bar-Shalom+Li+Kirubarajan:2001} as the dynamic model and with noisy location measurements. In the simulation, the process noise $\mathbf{q}_t$ was set to be zero with probability $0.8$ at every step $t$. The state $\mathbf{x}_t$ has the location $(x_{t,1}, \,x_{t,2})$ and the velocities $(x_{t,3}, \, x_{t,4})$. The measurement model matrix and the measurement noise covariance are
	\begin{equation*}
		\mathbf{H}_t = 
		\begin{bmatrix}
			1 & 0 & 0 &0 \\
			0 & 1 & 0 & 0
		\end{bmatrix}, \quad
		\mathbf{R}_t =
		\begin{bmatrix} 
			\sigma^2   & 0 \\
			0     &\sigma^2  
		\end{bmatrix}.
	\end{equation*}
	The transition matrix and the process noise covariance are 
	\begin{equation*}
		\mathbf{A}_t = 
		\begin{bmatrix}
			1 & 0 & \triangle t &0 \\
			0 & 1 & 0 &\triangle t \\
			0 & 0 & 1 & 0 \\
			0 & 0 & 0 & 1 
		\end{bmatrix},
		\mathbf{Q}_t = q_c
		\begin{bmatrix} 
			\frac{\Delta t^3 }{3} & 0 & \frac{\Delta t^2 }{2} & 0 \\
			0 & \frac{\Delta t^3 }{3}  & 0 & \frac{\Delta t^2 }{2}  \\
			\frac{\Delta t^2 }{2}  & 0 &  {\Delta t }  & 0 \\
			0 &  \frac{\Delta t^2 }{2}  & 0 & {\Delta t }
		\end{bmatrix}.
	\end{equation*}
	We have $\Delta t = 0.1$, $q_c = 0.5$, $\sigma = 0.3$, $T = 100$, $\mathbf{m}_1 =\begin{bmatrix}0.1& 0& 0.1& 0 \end{bmatrix}^\top$, and $\mathbf{P}_1$ is an identity matrix. We set the matrix $\mathbf{G}_{g,t}$ to an identity matrix and use the parameters $\gamma = 1$, $\mu=1$, and $K_{\max} = 50$.

	We define the estimation error as
	\begin{equation*}
		\mathbf{x}_{\text{err}}= \frac{\sum_{t=1}^{T} \| \mathbf{x}_t^{(k)} - \mathbf{x}_t^{\text{true}}  \|_2}{\sum_{t=1}^{T} \|\mathbf{x}_t^{\text{true}}\|_2},
	\end{equation*}
	where $\mathbf{x}_t^{\text{true}}$ is the ground truth. The estimation results are plotted in Fig.~\ref{fig:linear_car}, where the circles denote the noisy measurements and the blue dash line denotes the true state. As we can seen, the KS-mADMM estimate (black line) is much closer to the ground truth than the KS estimate (red dash line), which is also reflected by a lower error. 
	\begin{figure}[!htb]  
		\centerline{\includegraphics[width=0.9\columnwidth]
			{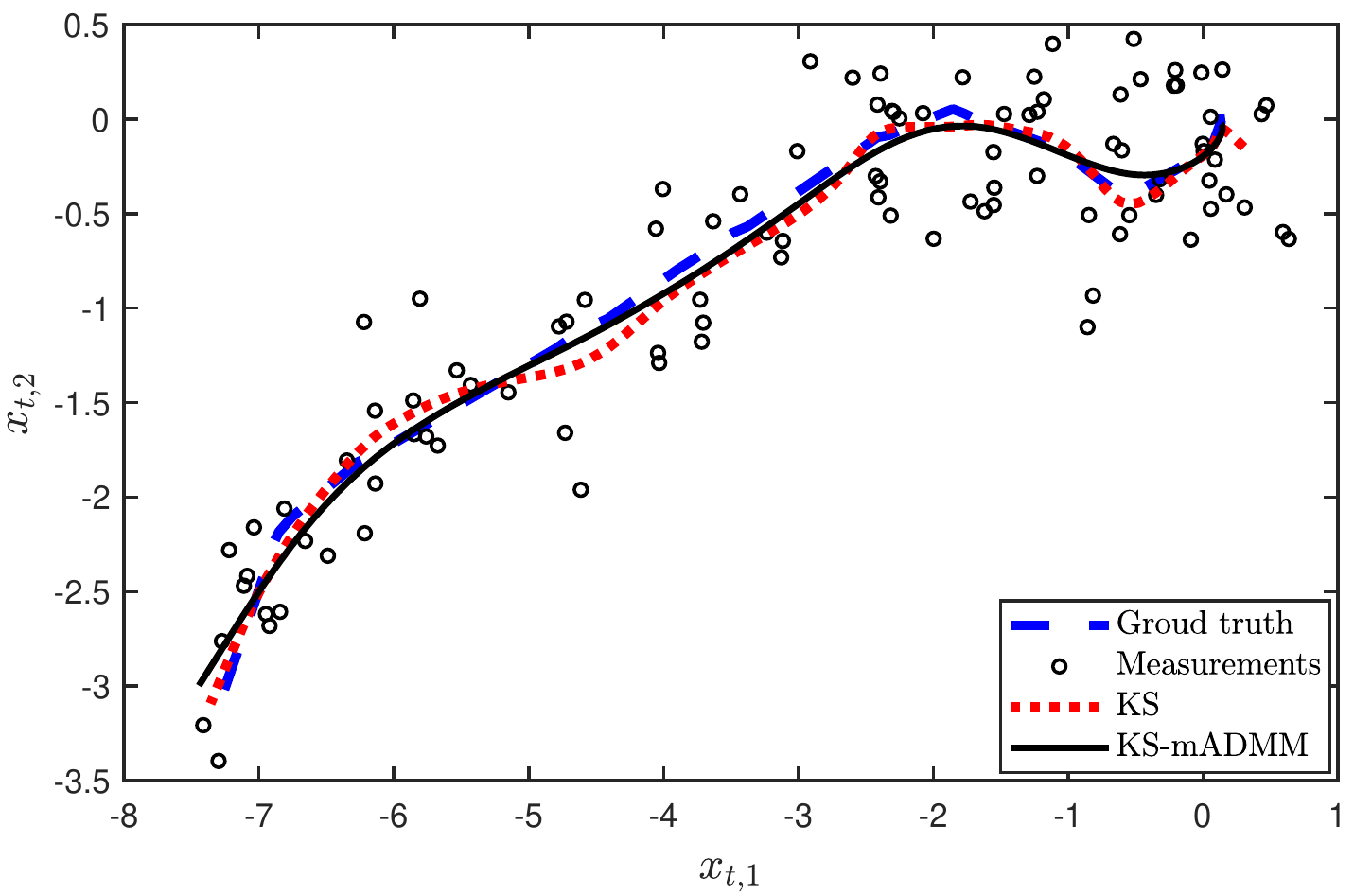}}
		\caption{Signals, measurements, and the estimates in the linear tracking problem. The values of $\mathbf{x}_{\text{err}}$ are $0.103$ and $0.072$ in KS and KS-mADMM, respectively.}
		\label{fig:linear_car}
	\end{figure}

	Recall that the difference in batch and recursive ADMM running time is dominated by the $\mathbf{x}_{1:T}$-subproblem. Fig.~\ref{fig:linear_large} demonstrates how the running time (sec) grows when $T$ is increasing. Despite being mathematically equivalent, \mbox{mADMM} and \mbox{KS-mADMM}, have very different running times. The running times of mADMM and prox-ADMM have a similar growth rate whereas  \mbox{KS-mADMM} has a growth rate that resembles a plain Kalman smoother.
	Due to limited memory, we cannot report the results of the batch estimation methods (prox-ADMM and mADMM) when $T > 10^4$. At $T = 10^4$, the running times of KS, KS-mADMM, prox-ADMM, and mADMM, were $0.34$s, $1.92$s, $6284$s and $9646$s, respectively. The proposed method is computationally inexpensive, which makes it suitable for solving real-world applications, such as the marine vessel tracking in Section~\ref{sec:linear_Marine}.
	\begin{figure}[!thb]  
		\centerline{\includegraphics[width=0.9\columnwidth]
			{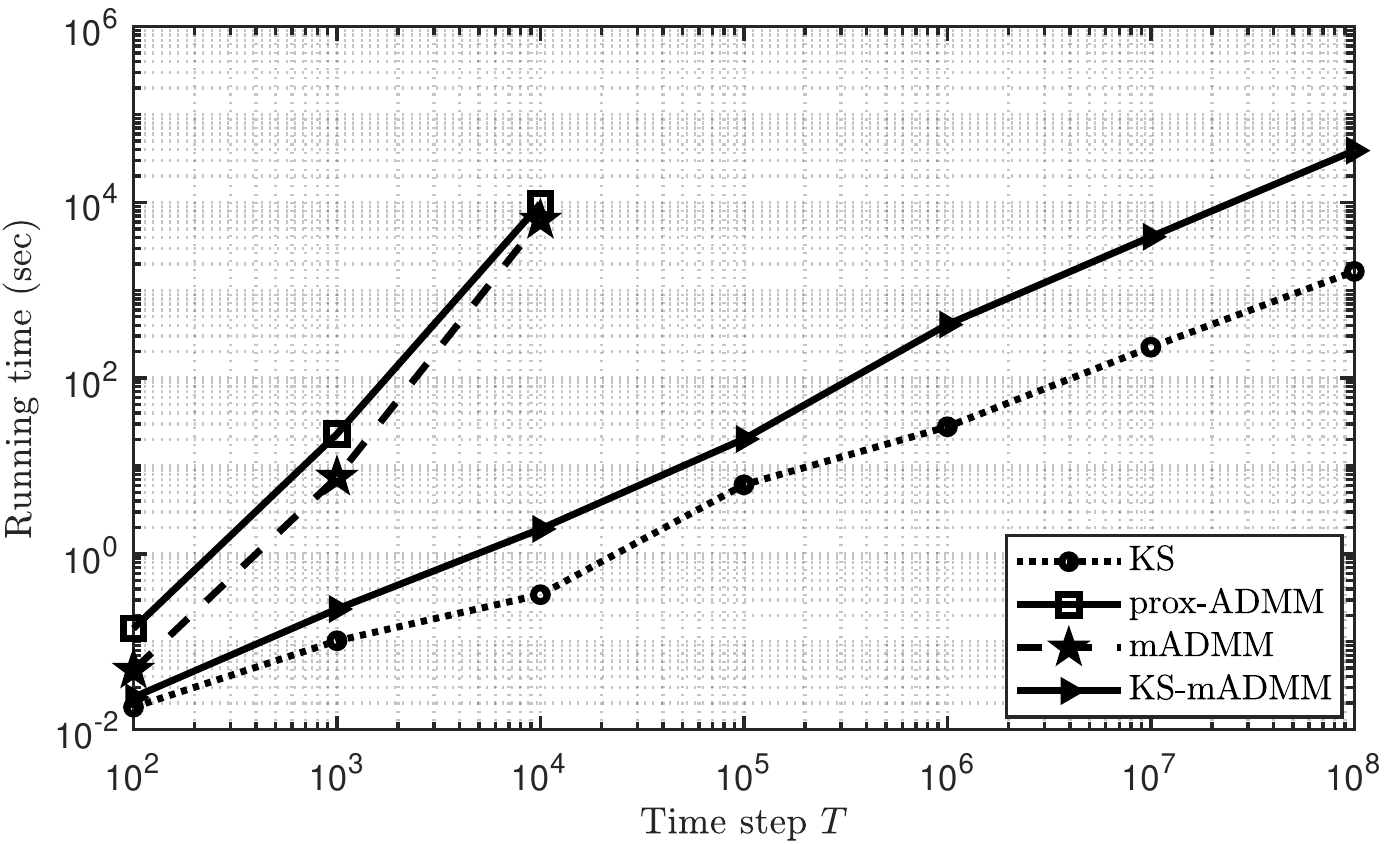}}
		\caption{Comparison of the running times in the linear car tracking example as function of the number of time steps.}
		\label{fig:linear_large}
	\end{figure}

	\subsection{Multi-Sensor Range Measurement Problems}
	\label{sec:nonlinear_car}
	In this experiment, we consider a multi-sensor range measurement problem where we have short periods of movement with regular stops. This problem frequently appears in many surveillance systems \cite{Bar-Shalom+Li+Kirubarajan:2001,Rogez2014Monocular}. The state $\mathbf{x}_t$ contains the position $(x_{t,1},x_{t,2})$ and the velocities $(x_{t,3},x_{t,4})$. The measurement dynamic model for sensor $n \in \left\{1,2,3\right\}$ is given by
	\begin{equation}
		\mathbf{h}_t^n(\mathbf{x}_t) 
		= 	
		\sqrt{(x_{t,2} - s_y^n)^2 + (x_{t,1} - s_x^n)^2}, \notag
	\end{equation}
	where $(s_x^n,s_y^n)$ is the position of the sensor $n$. The transition function $\mathbf{a}_t(\mathbf{x}_{t-1}) $ is 
	\begin{equation}
		\mathbf{a}_t(\mathbf{x}_{t-1})  =
		\begin{bmatrix}
			1 & 0 & \triangle t &0 \\
			0 & 1 & 0 &\triangle t \\
			0 & 0 & 1 & 0 \\
			0 & 0 & 0 & 1 
		\end{bmatrix} \, \mathbf{x}_{t-1}.	
	\end{equation}
	The covariances are $\mathbf{R}_t = \operatorname{diag}(0.2^2, 0.2^2)$, and $\mathbf{Q}_t = \operatorname{diag}(0.01, 0.01, 0.1, 0.1)$. We set $\Delta t = 0.1$, $T = 60$, $(s_x^1,\, s_y^1) = (0,-0.5)$, $(s_x^2,\, s_y^2) = (0.5,0.6)$, $(s_x^3,s_y^3) = (0.5,0.6)$, $\mathbf{m}_1=\begin{bmatrix}0 & 0 & 0& 0 \end{bmatrix}^\top$, and $\mathbf{P}_1 = \mathbf{I}/10$. We assume the target has many stops, which means the velocities $x_{t,3}$, $x_{t,4}$ are sparse. We also set
	$
	\mathbf{G}_{g,t} =
	\begin{bmatrix} 
		\mathbf{0} & \mathbf{I}
	\end{bmatrix}, 
	$
	and use the parameters $\gamma = 1$, $\mu=1$, $K_{\max} = 50$, and $I_{\max} = 5$.
	We plot the velocity variable $x_{t,3}$ corresponding to the time step $t$ in Fig.~\ref{fig:nonlinear_ct}, which indicates that our method (black line) can generate much more sparse results than the IEKS estimate (red dash line). 
	\begin{figure}[!htb]  
		\centerline{\includegraphics[width=0.9\columnwidth]
			{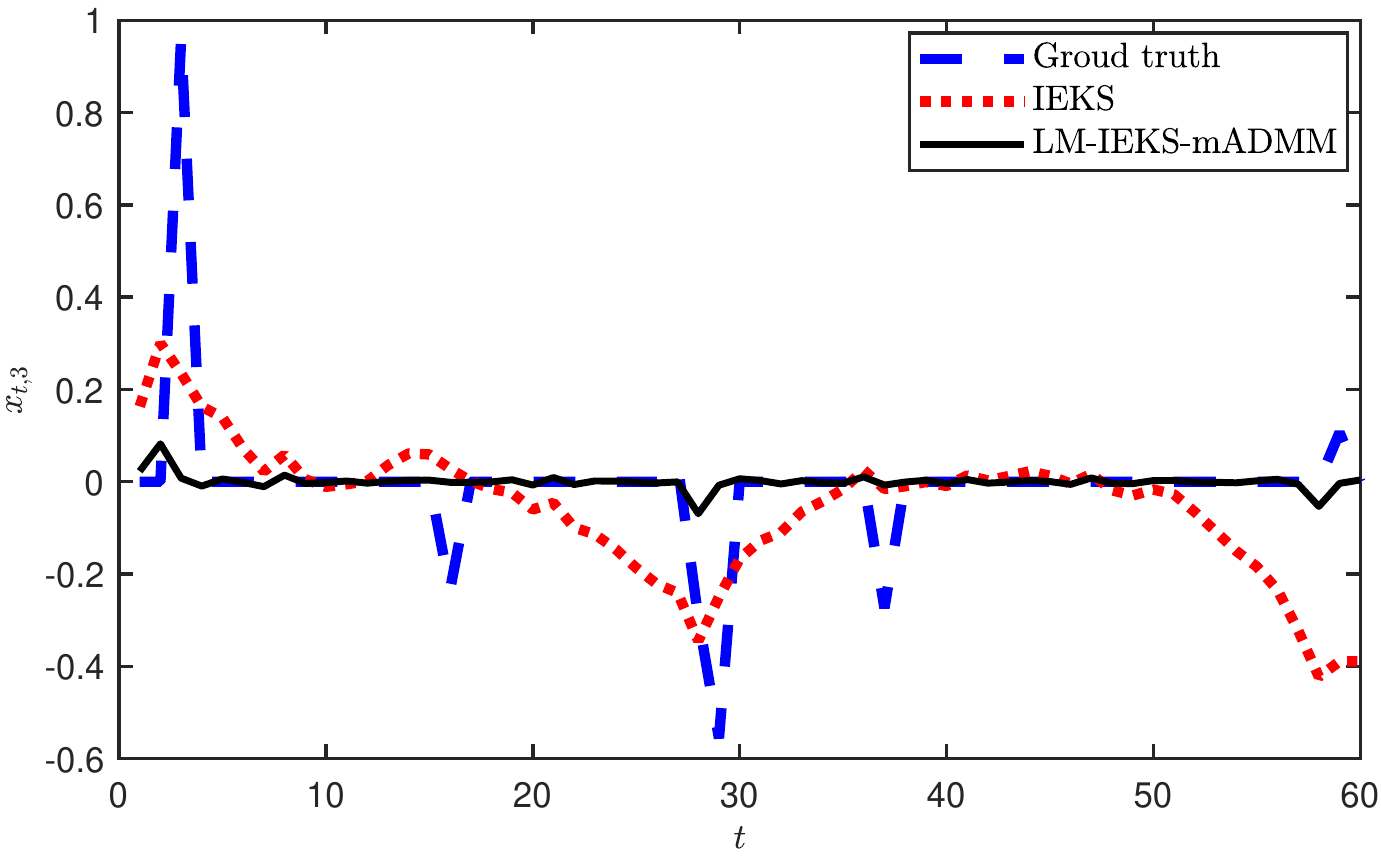}}
		\caption{The estimated trajectory in the nonlinear system. The relative errors are $0.53$ and $0.46$ generated by IEKS and LM-IEKS-ADMM.}
		\label{fig:nonlinear_ct}
	\end{figure}
	
	Fig.~\ref{fig:nonlinear_compared} shows the relative error $\mathbf{x}_{\text{err}}$ as a function of the iteration number. The values of $\mathbf{x}_{\text{err}}$ generated by the regularisation methods are below those generated by iterated extended Kalman smoother (IEKS) \cite{simo2013Bayesian}. It also shows that the \mbox{GN-mADMM}, \mbox{GN-IEKS-mADMM}, \mbox{LM-mADMM} and \mbox{LM-IEKS-mADMM} can find the optimal values in around $50$ iterations. IEKS is the fastest method, but the relative error is highest due to lack of the sparsity prior (i.e., $\mu = 0$). \mbox{GN-mADMM} and \mbox{GN-IEKS-mADMM} have the same convergence results (as they are equivalent), while the latter uses the less running time. Similarly, \mbox{LM-mADMM} and \mbox{LM-IEKS-mADMM} have the same convergence results, but \mbox{LM-IEKS-mADMM} needs less time to obtain the result than \mbox{LM-mADMM}. When the number of time steps $T$ is moderate, all the running time are acceptable. {But when $T$ is extremely large, the proposed methods provide a massive advantage.} 
	\begin{figure}[htb]  
		\centerline{\includegraphics[width=0.9\columnwidth]
			{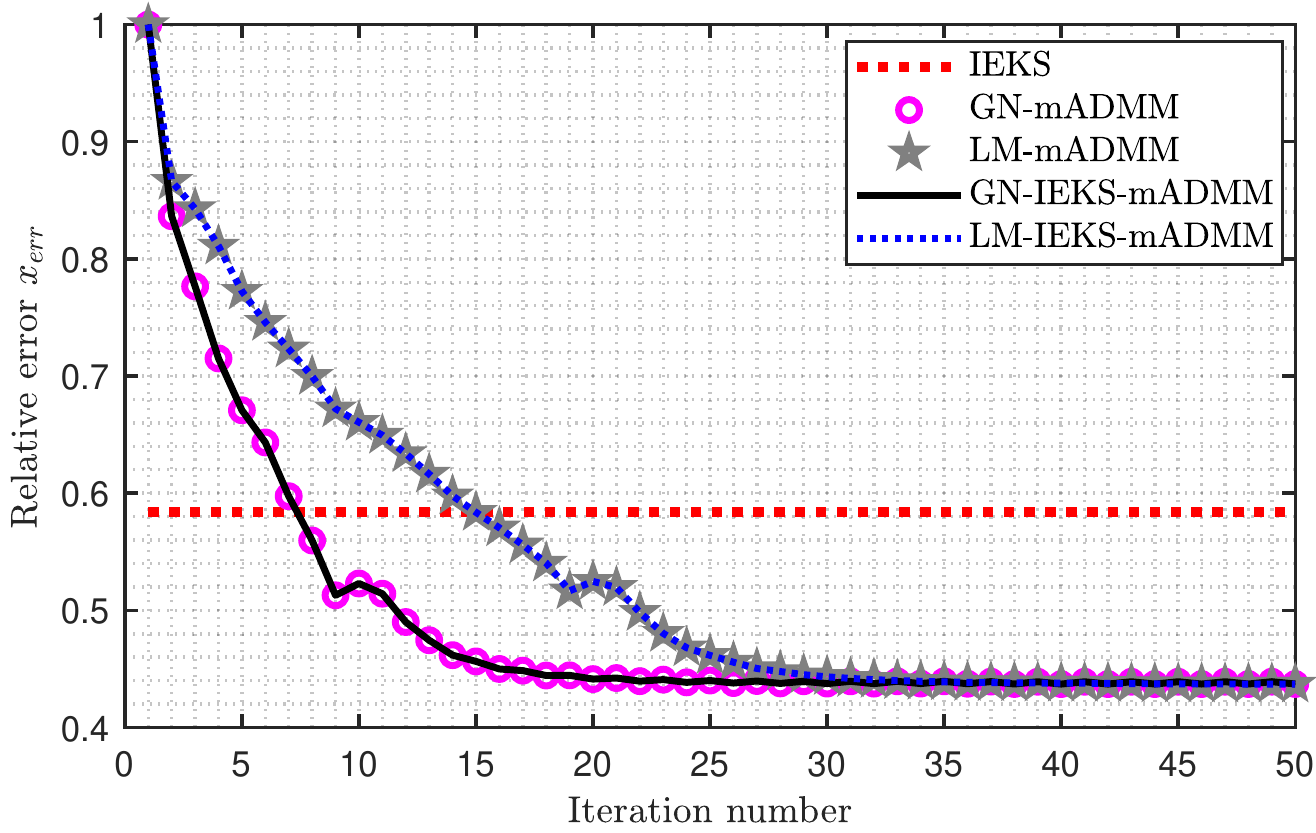}}
		\caption{Relative error $\mathbf{x}_{\text{err}}$ versus iteration number.}
		\label{fig:nonlinear_compared}
	\end{figure}

	Fig.~\ref{fig:nonlinear_large} demonstrates how the running time (sec) grows when $T$ is increasing. The proposed methods are compared with the state-of-the-art methods, including the proximal ADMM (prox-ADMM) \cite{Wright2006Numerical}, mADMM \cite{Boyd2011admm}, and IEKS \cite{simo2013Bayesian}. Despite being mathematically equivalent, \mbox{GN-mADMM} and \mbox{GN-IEKS-mADMM}, \mbox{LM-mADMM} and \mbox{LM-IEKS-mADMM}, have very different running times. \mbox{GN-IEKS-mADMM} and \mbox{LM-IEKS-mADMM} are more efficient than the batch methods. Due to limited memory, we cannot report the results of the batch methods when $T > 10^4$. It is reasonable to conclude that in general, the proposed methods are competitive for extremely large-scale tracking and estimation problems. The proposed approaches are computationally inexpensive, which makes them suitable for solving real-world applications, such as ship trajectory-tracking in the next section.
	\begin{figure}[!htb]  
		\centerline{\includegraphics[width=0.9\columnwidth]
			{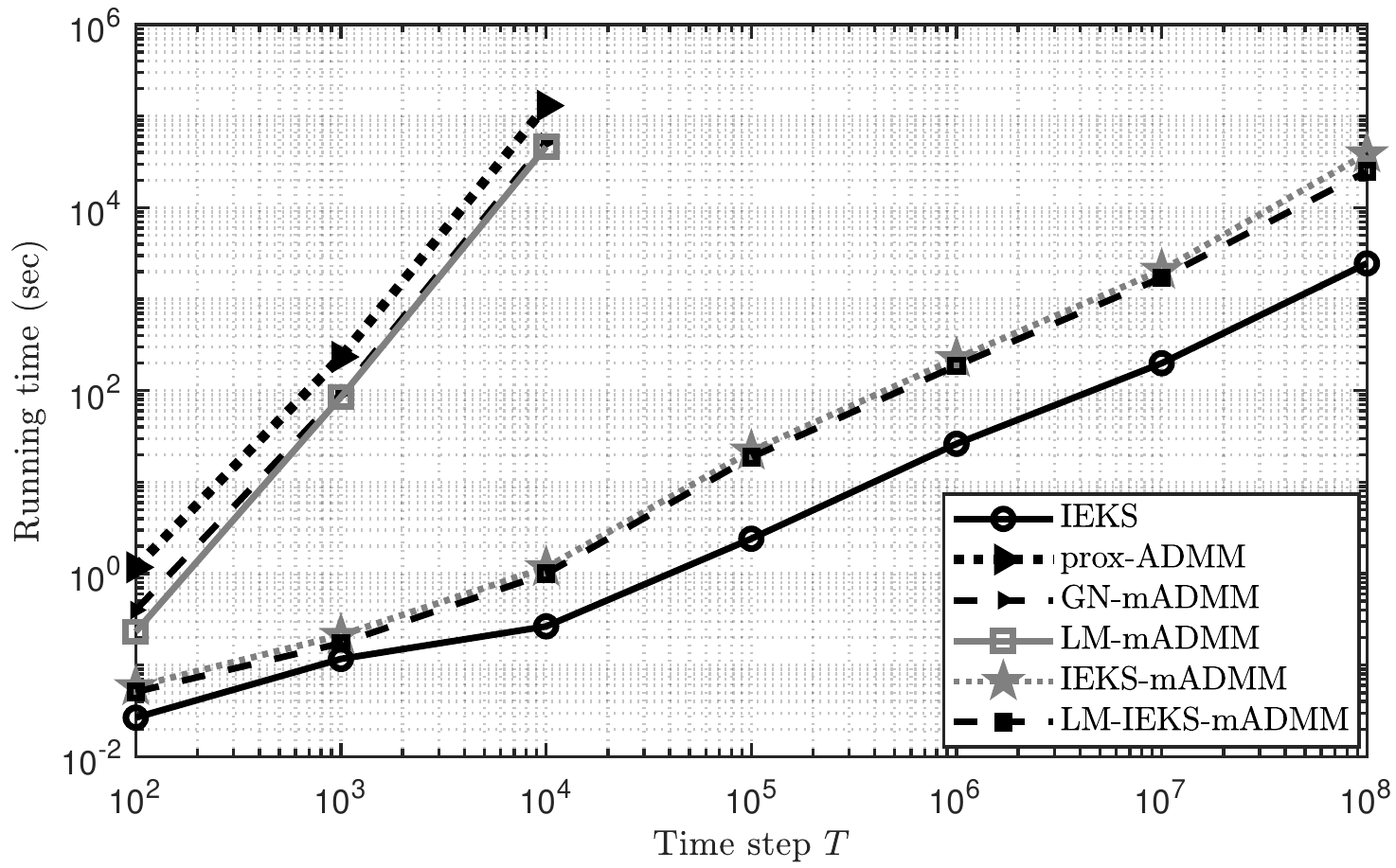}}
		\caption{Comparison of the running times in the range measurement example as function of the number of time steps (from $10^2$ to $10^8$).}
		\label{fig:nonlinear_large}
	\end{figure}

	\subsection{Marine Vessel Tracking}
	\label{sec:linear_Marine}
	
	In this experiment, we utilise the Wiener velocity model \cite{Bar-Shalom+Li+Kirubarajan:2001} with a sparse noise assumption to track a marine vessel trajectory. The latitude, longitude, speed, and course of the vessel have been captured by automatic identification system (AIS) equipment, collected by Danish Maritime Authority. Similar applications can be found in \cite{Ahmad2016Intent,McCall2016Integral}. The state of the ship is measured at time intervals of $1$ minute. Matrices $\mathbf{H}_t$,  $\mathbf{A}_t$, $\mathbf{Q}_t$, and $\mathbf{R}_t$ are the same with the settings in Section \ref{sec:linear_car} with $\Delta t = 1$, $q_c = 1$, $\sigma = 0.3$, $T = 100$, $\mathbf{m}_1=\begin{bmatrix}0.1 & 0.1 & 0 & 0 \end{bmatrix}^\top$, and $\mathbf{P}_1 = 100\mathbf{I}$. We assume the process noise $\mathbf{q}_t$ is sparse, and set $\mathbf{G}_{g,t}$ to an identity matrix and use the parameters $\gamma = 1$, $\mu=1$, and $K_{\max} = 100$. The measurement data consists of $100$ time points of the vessel locations. 
	
	Our method obtains the position (latitude and longitude) estimates as shown in Fig.~\ref{fig:linear_ship}. Fig.~\ref{fig:ship_error} shows that our method has sparser process noise than estimated by the Kalman smoother (KS) \cite{simo2013Bayesian}. We then highlight the computation advantage of our method. The difference in running time is dominated by the $\mathbf{x}_{1:T}$-subproblem. The running times of KS, prox-ADMM, mADMM, and KS-mADMM, were $0.34$s, $174$s, $172$s, and $5.63$s, respectively. The running times of mADMM and prox-ADMM are similar whereas \mbox{KS-mADMM} has a smaller running time that resembles the plain Kalman smoother.

	\begin{figure}[!htb]  
		\centerline{\includegraphics[width=0.9\columnwidth]
			{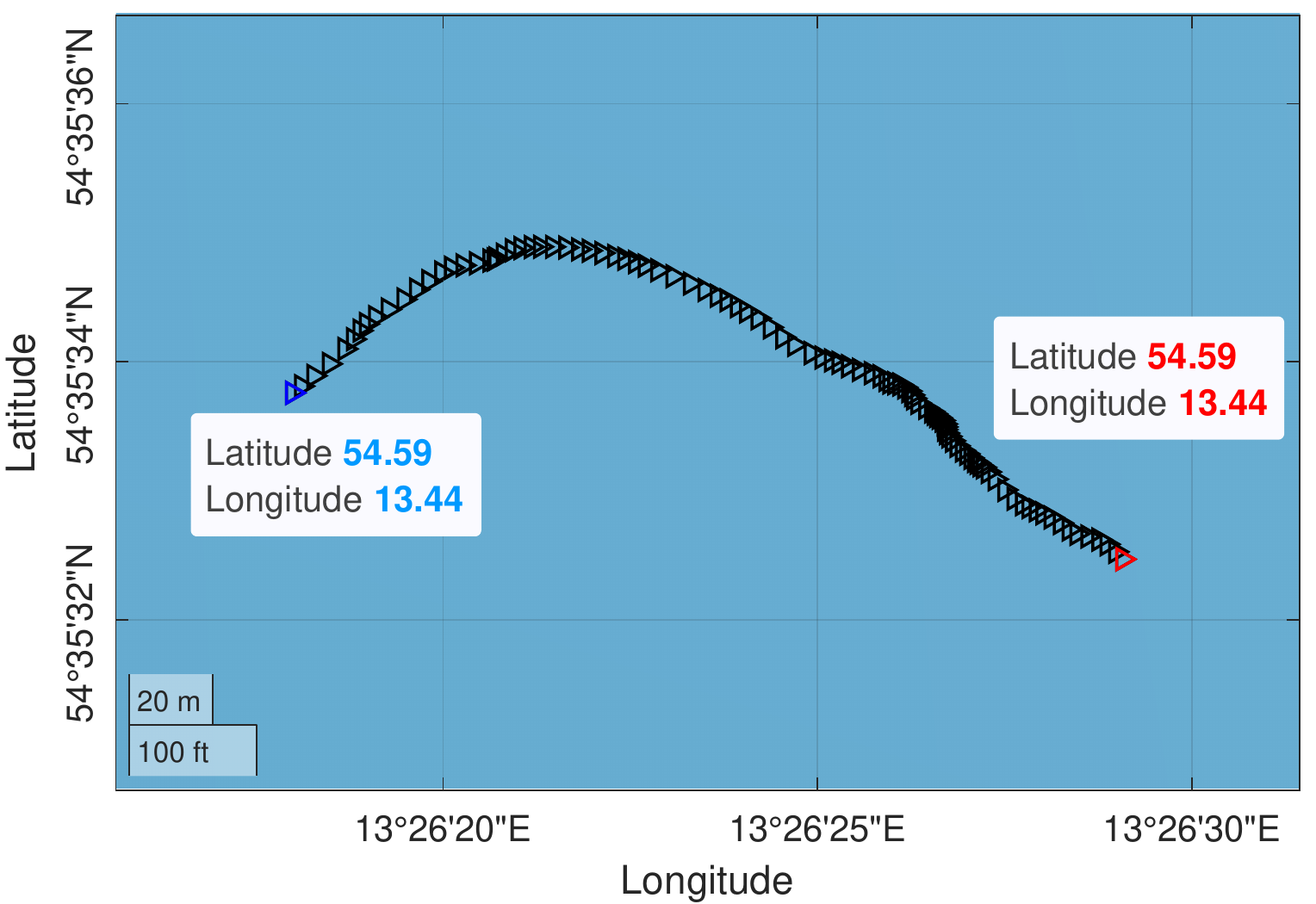}}
		\caption{The position (black markers) estimated by KS-mADMM. The starting coordinate is denoted blue marker, and the ending coordinate is red marker. Contains data from the Danish Maritime Authority that is used in accordance with the conditions for the use of Danish public data.}
		\label{fig:linear_ship}
	\end{figure}
	\begin{figure}[!htb]  
		\centerline{\includegraphics[width=0.9\columnwidth]
			{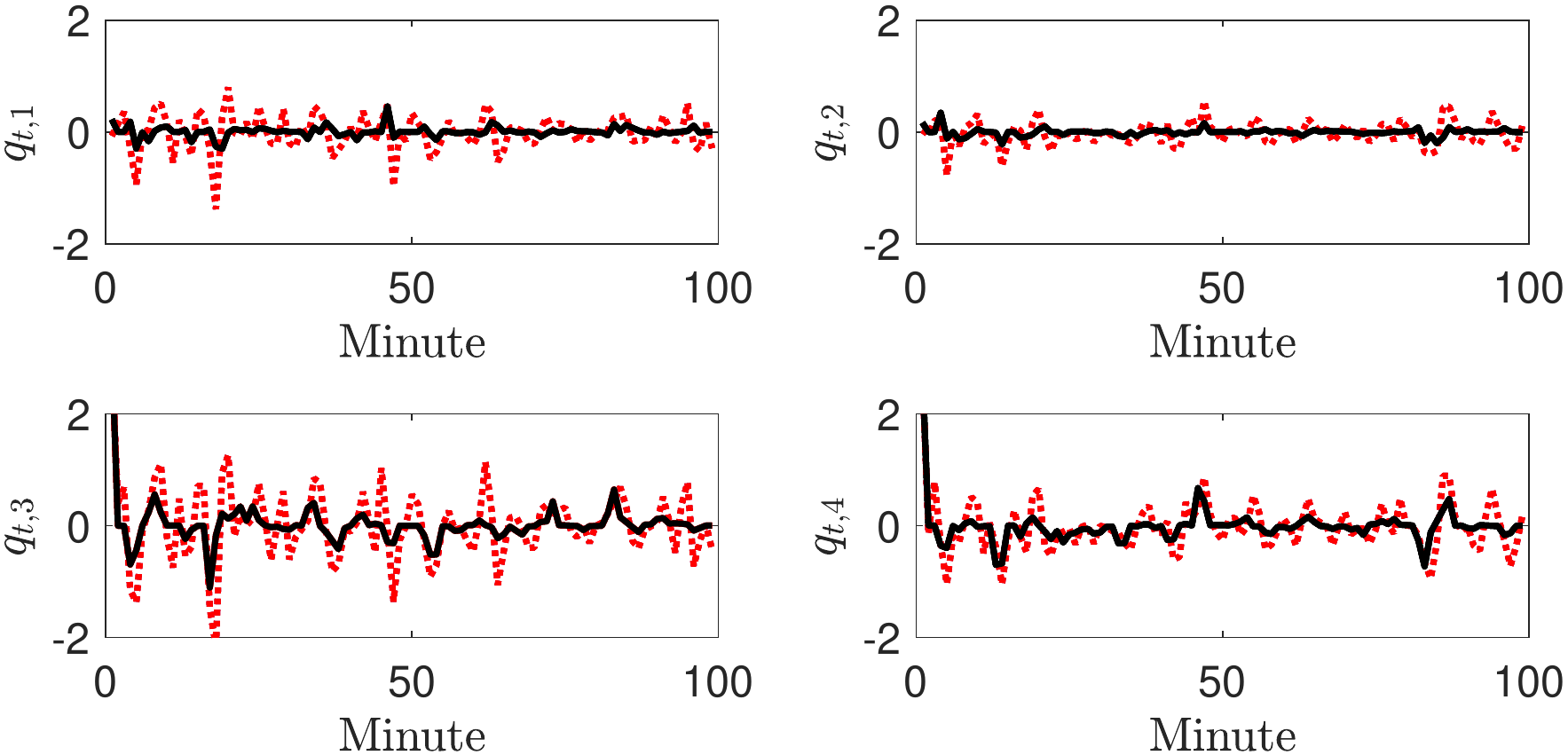}}
		\caption{The process noise estimated by KS-mADMM (black line) and Kalman smoother (red dash line).}
		\label{fig:ship_error}
	\end{figure}

	\subsection{Autonomous Vehicle Tracking}
	\label{sec:urban_car}
	To further show how our methods can speed up larger scale real-world problems, we apply GN-IEKS-mADMM to a vehicle tracking problem using real-world data. Global positioning system (GPS) data was collected in urban streets and roads around Helsinki, Tuusula, and Vantaa, Finland \cite{Jouni2011GPS}. The urban environment contained many stops to traffic lights, crossings, turns, and various other situations. We ran the experiment using a coordinate turn model \cite{simo2013Bayesian}, where the state at time step $t$ had the positions $(x_{t,1}, \,x_{t,2})$, the velocities $(x_{t,3}, \, x_{t,4})$, and the angular velocity $x_{t,5}$. The number of time points $T$ was $6865$. 
	We use the parameters $\gamma = 0.1$, $\mu = 1$, $K_{\max} = 300$, $I_{\max} =5$, $\mathbf{m}_1 = \begin{bmatrix}4.5 &13.5 & 0 & 0& 0 \end{bmatrix}^\top$, and $\mathbf{P}_1= \operatorname{diag}(50,50,50,50,0.01)$. We utilised the matrix
	\begin{equation*}
		\mathbf{G}_{g,t} = 
		\begin{bmatrix}
			0 & 0  & 1  &0 & 0\\
			0  & 0  & 0 & 1 & 0 \\
			0  & 0  & 0 & 0 & 1
		\end{bmatrix}, 
	\end{equation*}
	to enforce the sparsity of the velocities and the angular velocity. 
	
	The plot in Fig.~\ref{fig:car_map} demonstrates the path (black line) generated by our method. The running time of IEKS \cite{Aravkin2017Generalized}, GN-mADMM, and GN-IEKS-mADMM were $22$s, $13520$s and $2704$s, respectively. As we expected, although IEKS is fastest, the $L_2$-penalised regularisation methods push more of the velocities and the angle to zero, which is shown in Fig.~\ref{fig:car_sparsity}. The IEKS estimate has many large peaks that appear as a result of large residuals, and GN-IEKS-mADMM has more sparse results.
	\begin{figure}[!htb]  
		\centerline{\includegraphics[width=0.9\columnwidth]
			{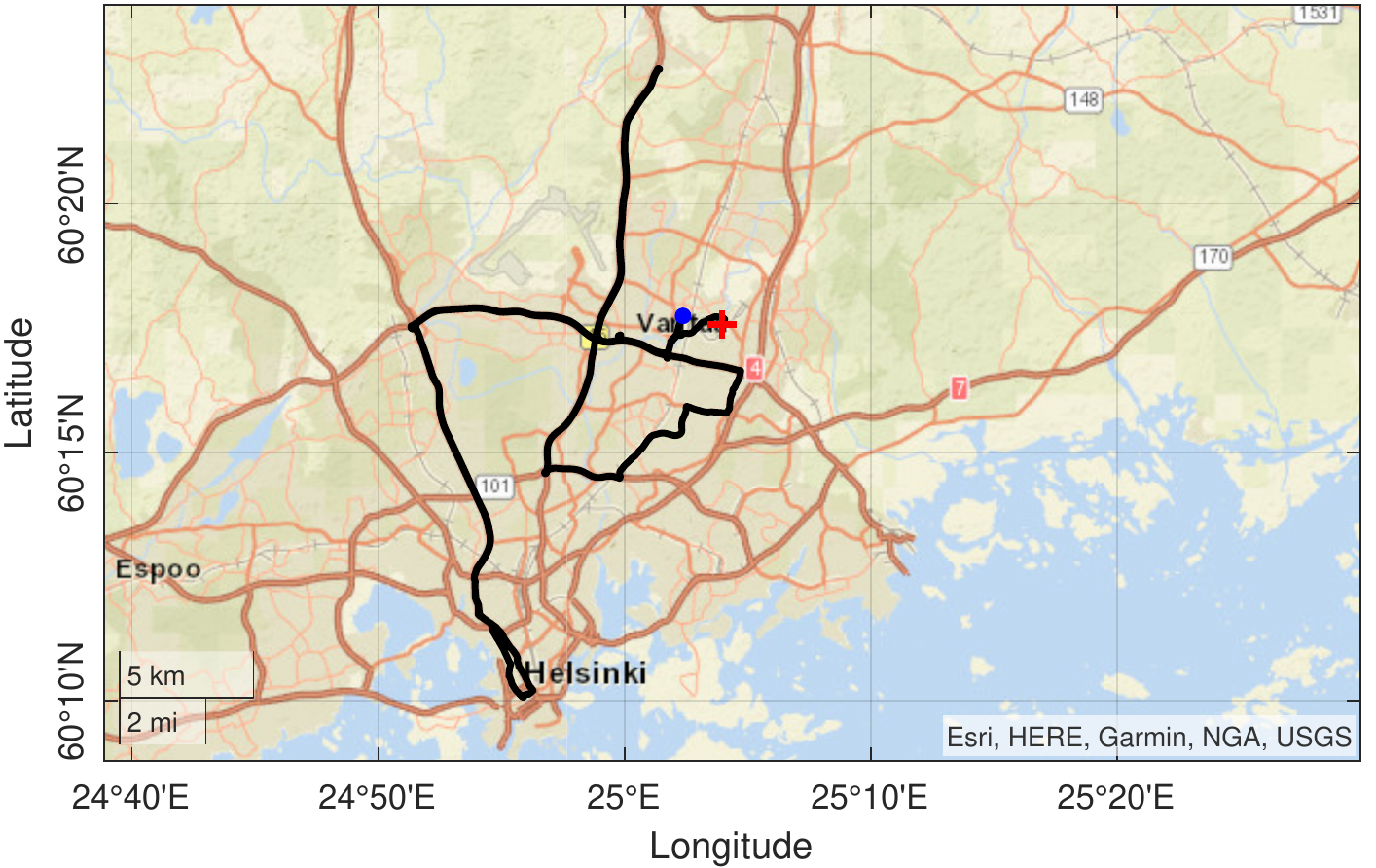}}
		\caption{The path tracking (black line) generated by GN-IEKS-mADMM. The starting position is blue point, and the ending position is red cross.}
		\label{fig:car_map}
	\end{figure}
	\begin{figure}[!thb]  
		\centerline{\includegraphics[width=0.9\columnwidth]
			{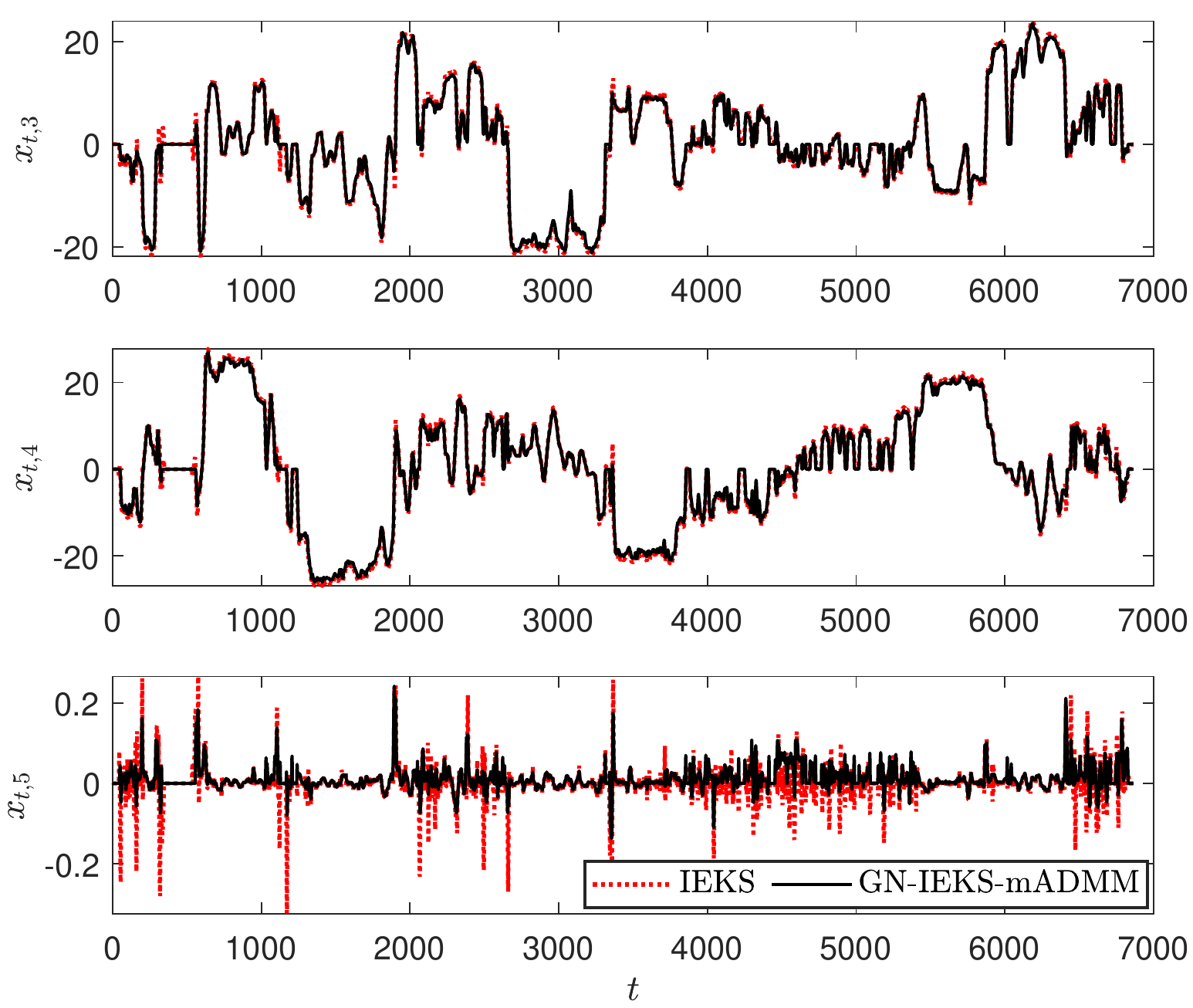}}
		\caption{The estimated velocities and angular velocities generated by IEKS (red dash line) and the proposed method (black line). }
		\label{fig:car_sparsity}
	\end{figure}
	
	\subsection{Audio Signal Restoration}
	\label{sec:audio}
	
	The proposed technique can be readily applied to the problem of noise reduction in audio signals. We adopt a Gabor regression model \cite{wolfe2004bayesian}: 
	\begin{equation} y(\tau) = \sum_{m=0}^{M/2} \sum_{n=0}^{N-1}  c_{m,n} g_{m,n}(\tau) + r(\tau), ~~~~\tau = 0, \dots, T-1\notag,
	\end{equation}
	where signals are represented as a weighted sum of Gabor atoms $g_{m,n}(\tau) =  ~w_n(\tau) \exp \left(2 \pi \mathrm{i} \frac{m}{M} \tau\right)$. Terms $w_n(\tau)$ correspond to a window function with bounded support centred at time instants $\tau_n$ (windows are placed so that the time axis is with tiled evenly). Sparsity is promoted through the $L_2^1$ pair-wise grouping pattern described in Section~\ref{sec:formulation}: $\sum_{m,n}\mu_{m,n}\| c_{m,n}\|_2$. The real representation of complex coefficients $c_{m,n}$ used in \cite{wolfe2004bayesian} is adopted. This batch problem is restated in terms of a state-space model: signal $\mathbf{y}$ is separated into $P$ chunks $\mathbf{y}_{t}$ of length $L$ and state vectors $\mathbf{x}_t = [\mathbf{c}_{2(t-1)}; \mathbf{c}_{2t-1}; \mathbf{c}_{2t}]^\top$ are defined, $\mathbf{c}_{t}$ being the subvector associated to each frame. Let $\mathbf{H}_0$ be a matrix containing the non-zero values of the Gabor basis functions $\mathbf{g}_{0,0}, \dots, \mathbf{g}_{M/2,0}$ as columns. Thus, atoms in subsequent frames are time-shifted replicas of this basic set and $\|\mathbf{y} - \mathbf{D}\mathbf{c}\|^2$ ($\mathbf{D}$ a \emph{dictionary} matrix containing all atoms) can be replaced by $\sum_{t=1}^P \|\mathbf{y}_t - \mathbf{H_\ast} \mathbf{x}_t\|^2$ $+\sum_{t=1}^P \| \mathbf{x}_t -  \mathbf{A}_t\mathbf{x}_{t-1}\|^2$, with $\mathbf{H}_\ast = \begin{bmatrix}\mathbf{H}_u & \mathbf{H}_0  &\mathbf{H}_\ell\end{bmatrix}$ and 
	$\mathbf{A}_t = \begin{bmatrix} \mathbf{0} & \mathbf{0} &  \mathbf{0}; \quad \mathbf{0} & \mathbf{0} & \mathbf{0}; \quad \mathbf{I} &  \mathbf{0} & \mathbf{0}  \end{bmatrix}$. Terms $\mathbf{H}_u, \mathbf{H}_\ell$ are truncated versions of $\mathbf{H}_0$ corresponding to the contribution of the adjacent overlapping frames. 
	
	The algorithm is tested on a $\sim$3~second long glockenspiel excerpt sampled at 22050~[Hz] and contaminated with artificial background noise with signal-to-noise ratio (SNR) $5 \text{dB}$. Experiments are carried out in an Intel Core i7 @ 2.50GHz, 16 GB RAM, with parameters $\gamma=5$, $\mu=2.6$, and $K_{\max} = 500$ and a window length $L = 512$. Kalman gain matrices are precomputed. Reflecting the power spectrum of typical audio signals, which decays with frequency, penalisation is made frequency-dependent by setting $\mu_{m,n} = \mu / f(m)$, with $f(m)$ a decreasing modulating function (e.g., a Butterworth filter gain), in a similar fashion to \cite{fevotte2007sparse}. Coefficients are initialised at zero.The average output SNR is $12.4$ with an average running time of $64.6$s in 20 realisations. Fig.~\ref{fig:audio_denoising} shows the visual reconstruction results.
	
	In comparison, Gibbs sampling schemes for models (e.g., \cite{wolfe2004bayesian}) yield noisier restorations with comparable computing times. We analysed the same example using the Gibbs sampler with 500 iterations, 250 burn-in period. Hyperparameters and initial values are chosen to ensure a fair comparison with the KS-mADMM method (unfavorable initialisation may induce longer convergence times). With a runtime of $\sim$180 seconds, the Gibbs algorithm yields an output SNR of $\sim$15 dB. The Perceptual Evaluation of Audio Quality (PEAQ) \cite{kabal2002examination}, a measure that incorporates psycho-acoustic criteria to asses audio signals, is adopted. The Objective Difference Grade (ODG) indicator derived from PEAQ is used to compare the reconstructions with respect to the clean reference signal, obtaining ODG = $-3.910$ for clean signal against noisy input, ODG = $-3.846$ for clean signal against Gibbs reconstruction, and ODG = $-3.637$ for clean signal against KS-mADMM reconstruction (the closer to 0, the better). Despite the lower SNR (12.4 dB), the KS-ADMM reconstruction sounds cleaner (i.e., has fewer audio artefacts) than its Gibbs counterpart, which is consistent with the ODG values obtained. Devising appropriate temporal evolution models for the audio synthesis coefficients over time and investigating self-adaptive schemes for the estimation of $\mu$ (here tuned empirically) are topics of future research.
	
	\begin{figure}[!htb]  
		\centerline{\includegraphics[width=0.85\columnwidth]
			{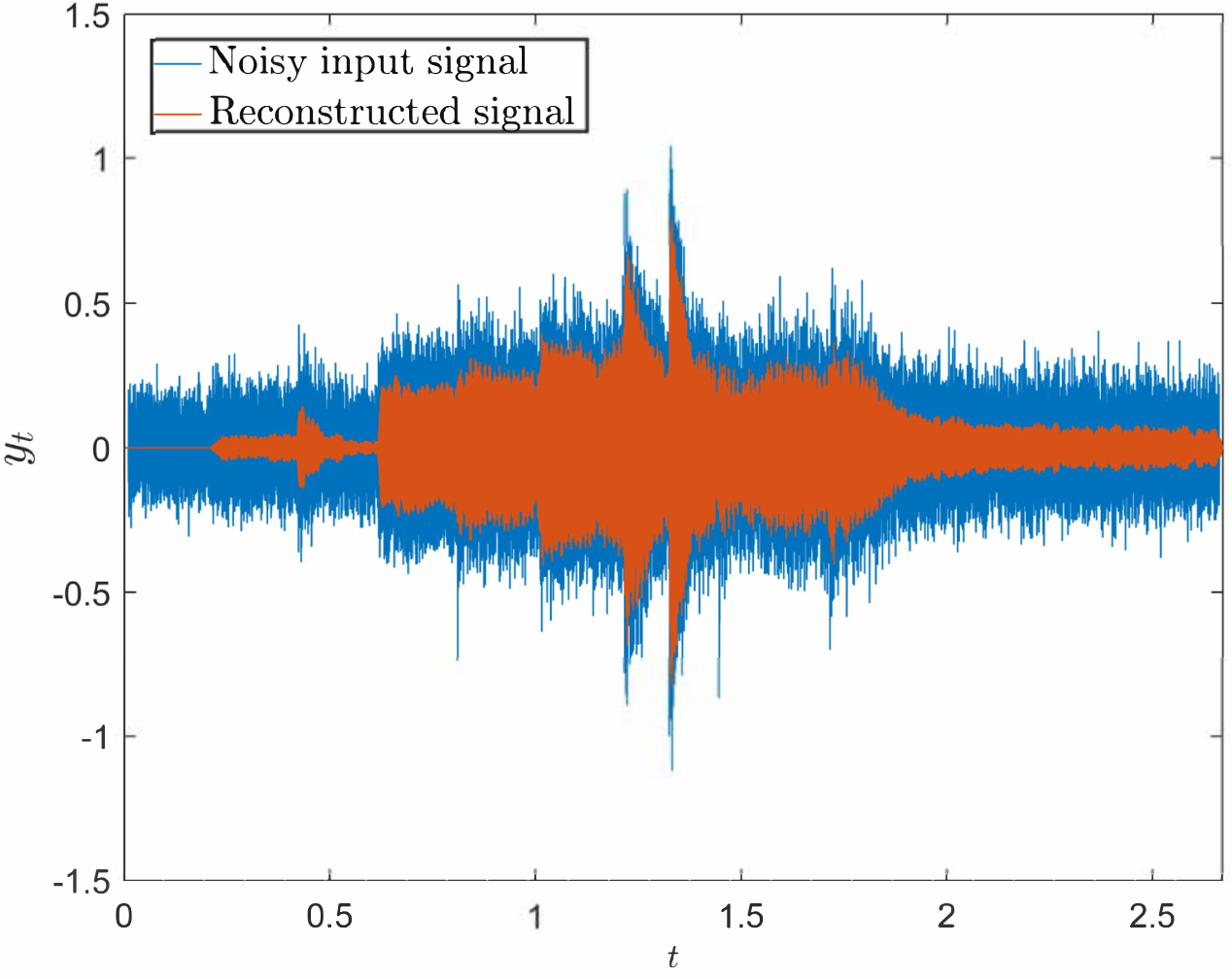}}
		\caption{Reconstructed glockenspiel excerpt.}
		\label{fig:audio_denoising}
	\end{figure}
	
	\section{Conclusion and Discussion}
	\label{sec:Conclusion}	
	In this paper, we have presented efficient smoothing-and-splitting methods for solving regularised autonomous tracking and state estimation problems. We formulated the problem as a generalised $L_2$-penalised dynamic group Lasso minimisation problem. The problem can be solved using batch methods when the number of time steps is moderate. For the case with a large number of time steps, new KS-mADMM, GN-IEKS-mADMM, and LM-IEKS-mADMM methods were developed. We also proved the convergence of the proposed methods. We applied the developed methods to simulated tracking, real-world tracking, and audio signal restoration problems, where methods resulted in improved localisation and estimation performance and significantly reduced computation load. 
	
	A disadvantage of the smoothing-and-splitting methods is that although the methods significantly improve the tracking and estimation performance, their reliability depends on user-defined penalty parameters (e.g., the parameter $\gamma$ in \eqref{eq:Lagrangian_general}). See \cite{Boyd2011admm,Wright2006Numerical} for further details on choosing the appropriate values of the parameters. The use of adaptive penalty parameters may improve the performance in dynamic systems, even while stronger conditions of convergence need to be guaranteed \cite{Zhou2019Adaptive}. It would be interesting to develop fully automated solvers with adaptive parameters. 
	The convergence and the convergence rate of our methods are based on Bayesian smoothers and ADMM, and we have established the convergence rate of the convex case. Possible future work includes discussing the convergence rate for nonconvex variants.
	
	Although we only consider the autonomous tracking and state estimation problems in this paper, it is possible to apply our framework to a wide class of control problems. For example, in linear optimal control problems, we could introduce splitting variables to decompose the nonsmooth terms and then use the Riccati equations to compute the subproblems arising in the optimal control problems \cite{Axehill2012Riccati}. In cooperative control of multiple target systems \cite{wen2020Sampled-Data,guo2020Asymptotic,guo2019Fuel-Efficient}, we can consider a reformulation of dynamic models of group targets into classes with different characteristics. Based on the framework, we address the subproblems in implementing optimisation-based methods such as receding horizon methods \cite{guo2020Asymptotic}. The proposed framework can be extended to other variable splitting methods \cite{Splitting2017book} as well as other recursive smoothers \cite{simo2013Bayesian}. Future work also includes developing other variants, for example, sigma-point based variable splitting methods.

	\appendices
	
	\section{Proof of Lemma~\ref{lemma:s_sequence}}
	\label{pf:s_sequence}
	
	For proving Lemma~\ref{lemma:s_sequence}, we define $\bm{\zeta} = \begin{bmatrix} \mathbf{x}& \mathbf{w}\end{bmatrix}^\top$ and then write variables $\mathbf{x}$ and $\mathbf{w}$ into the function $\Xi(\bm{\zeta})$, which is 
	\begin{equation} \label{eq:Xi_def}
		\Xi(\bm{\zeta}) =\frac{1}{2} \left \| {\mathbf{y}}- \mathbf{H} \mathbf{x} - \mathbf{e}\right \|_{\mathbf{R}^{-1}}^2  +\frac{1}{2} \left \| {\mathbf{m}}- \mathbf{\Phi} \mathbf{x} - \mathbf{b}\right \|_{\mathbf{Q}^{-1}}^2+\mu \left \| \mathbf{w}\right \|_2. 
	\end{equation}
	Using the optimality conditions of the subproblems in \eqref{eq:subproblem_dynamic}, we can write
	\begin{subequations}\label{eq:subproblem_linear_gradient} 
		\begin{align}
			\label{eq:z-primal-gradient}	
			&\Xi(\bm{\zeta})   - \Xi(\bm{\zeta}^{(k+1)}) + (\bm{\zeta} - \bm{\zeta}^{(k+1)})^\T
			\left(\begin{bmatrix}
				\mathbf{\Phi} & \mathbf{0}\\
				\mathbf{0}  & \mathbf{I} \end{bmatrix}^\top 
			\bm{\eta}^{(k)} \right.  \notag\\
			&\left. \quad + \gamma \begin{bmatrix}
				\mathbf{\Phi} & \mathbf{0}\\
				\mathbf{0}  & \mathbf{I} \end{bmatrix}^\top 
			\big(\bm{\zeta}^{(k+1)}- \begin{bmatrix}\mathbf{I}  \\ \mathbf{G} \end{bmatrix} \mathbf{v}^{(k)}  - \begin{bmatrix}
				\mathbf{d} \\
				\mathbf{0} \end{bmatrix}   \big) \right) 
			\geq \mathbf{0}, \\ 
			\label{eq:v-primal-gradient}
			&(\mathbf{v} - \mathbf{v}^{(k+1)})^\T
			\left(-\begin{bmatrix}
				\mathbf{I}  \\ \mathbf{G} \end{bmatrix}^\top \bm{\eta}^{(k)} \right. \notag\\ 
			&\left. \quad -\gamma \begin{bmatrix}
				\mathbf{I}  \\
				\mathbf{G}
			\end{bmatrix}^\top \big(\bm{\zeta}^{(k+1)}- \begin{bmatrix}
				\mathbf{I}  \\
				\mathbf{G}
			\end{bmatrix} \mathbf{v}^{(k+1)}  - \begin{bmatrix}
				\mathbf{d} \\
				\mathbf{0} \end{bmatrix}\big) \right) 
			\geq \mathbf{0}, \\ 
			\label{eq:eta-primal-gradient}
			&(\bm{\eta} -\bm{\eta}^{(k+1)} )^\T 
			\left(- \big(\begin{bmatrix}
				\mathbf{\Phi} & \mathbf{0}\\
				\mathbf{0}  & \mathbf{I} \end{bmatrix} \mathbf{z}^{(k+1)}
			- \begin{bmatrix}
				\mathbf{I}  \\
				\mathbf{G}
			\end{bmatrix} \mathbf{v}^{(k+1)} -\begin{bmatrix}
				\mathbf{d} \\
				\mathbf{0} \end{bmatrix}  \big) \right. \notag \\
			&\left.\quad + \frac{1}{\gamma}(\bm{\eta}^{(k+1)}  - \bm{\eta}^{(k)}  )   \right)
			\geq \mathbf{0}.
		\end{align}
	\end{subequations}
	For simplicity of notation, we also define
	\begin{equation}
		\begin{split}
			\begin{aligned}
				\bm{\xi} &= \begin{bmatrix}
					\bm{\zeta}&
					\mathbf{v}&
					\bm{\eta}
				\end{bmatrix}^\top,   \\
				F(\bm{\xi}) &= \begin{bmatrix}
					\begin{bmatrix}
						\mathbf{\Phi} & \mathbf{0}\\
						\mathbf{0}  & \mathbf{I} 
					\end{bmatrix}^\top \bm{\eta}^{(k+1)}   \\
					-\begin{bmatrix}
						\mathbf{I}  \\
						\mathbf{G}
					\end{bmatrix}^\top \bm{\eta}^{(k+1)} \\
					- \begin{bmatrix}
						\mathbf{\Phi} & \mathbf{0}\\
						\mathbf{0}  & \mathbf{I} 
					\end{bmatrix} \bm{\zeta}^{(k+1)}
					+ 
					\begin{bmatrix}
						\mathbf{I}  \\
						\mathbf{G}
					\end{bmatrix} \mathbf{v}^{(k+1)} +
					\begin{bmatrix}
						\mathbf{d}  \\
						\mathbf{0}
					\end{bmatrix} \\
				\end{bmatrix}.
			\end{aligned}
		\end{split}
	\end{equation} 
	We group all the variables $\bm{\zeta}$, $\mathbf{v}$, and $\bm{\eta}$ into a single vector $\bm{\xi}$, and then rewrite \eqref{eq:subproblem_linear_gradient} as follows:
	\begin{equation}\label{eq:subproblem_linear_gradient_3}
		\begin{split}
			\begin{aligned}
				&  \Xi(\bm{\zeta})   - \Xi(\bm{\zeta}^{(k+1)})   + 
				\begin{bmatrix}
					\bm{\xi} - \bm{\xi}^{(k+1)}   \\
				\end{bmatrix}^\T \\
				&\quad
				\left(
				F(\bm{\xi})
				+\gamma \begin{bmatrix}
					\begin{bmatrix}
						\mathbf{\Phi} & \mathbf{0}\\
						\mathbf{0}  & \mathbf{I} 
					\end{bmatrix}^\T  \\
					-\begin{bmatrix}
						\mathbf{I}  \\
						\mathbf{G}
					\end{bmatrix}^\T  \\
					\mathbf{0} \\
				\end{bmatrix} \begin{bmatrix}
					\mathbf{I}  \\
					\mathbf{G}
				\end{bmatrix}( \mathbf{v}^{(k)} - \mathbf{v}^{(k+1)}) \right. \\
				& \left. \qquad +
				\begin{bmatrix}
					&  \mathbf{0} &   \mathbf{0}\\
					&\gamma \begin{bmatrix}
						\mathbf{I}  \\
						\mathbf{G}
					\end{bmatrix}^\T &   \mathbf{0} \\
					& \mathbf{0}  &\frac{1}{\gamma}\mathbf{I} \\
				\end{bmatrix}
				\begin{bmatrix}
					\mathbf{v}^{(k+1)} - \mathbf{v}^{(k)}  \\
					\bm{\eta}^{(k+1)} -\bm{\eta}^{(k)} \\
				\end{bmatrix}
				\right) \geq  \mathbf{0}, \\ 
				&  \Xi(\bm{\zeta})   - \Xi(\bm{\zeta}^{(k+1)})   + 
				\begin{bmatrix}
					\bm{\xi} - \bm{\xi}^{(k+1)}   \\
				\end{bmatrix}^\T 
				F(\bm{\xi}) \\
				&\quad + \gamma  \begin{bmatrix}
					\bm{\xi} - \bm{\xi}^{(k+1)}   \\
				\end{bmatrix}^\T \begin{bmatrix}
					\begin{bmatrix}
						\mathbf{\Phi} & \mathbf{0}\\
						\mathbf{0}  & \mathbf{I} 
					\end{bmatrix}^\T  \\
					-\begin{bmatrix}
						\mathbf{I}  \\
						\mathbf{G}
					\end{bmatrix}^\T  \\
					\mathbf{0} \\
				\end{bmatrix} \begin{bmatrix}
					\mathbf{I}  \\
					\mathbf{G}
				\end{bmatrix} \big( \mathbf{v}^{(k)} - \mathbf{v}^{(k+1)}\big)  \\
				& \quad \geq 
				\begin{bmatrix}
					\mathbf{v} - \mathbf{v}^{(k+1)}  \\
					\bm{\eta} -\bm{\eta}^{(k+1)} \\
				\end{bmatrix}^\top
				\begin{bmatrix}
					&\gamma \begin{bmatrix}
						\mathbf{I}  \\
						\mathbf{G}
					\end{bmatrix}^\T &   \mathbf{0} \\
					& \mathbf{0}  &\frac{1}{\gamma}\mathbf{I} \\
				\end{bmatrix}
				\begin{bmatrix}
					\mathbf{v}^{(k+1)} - \mathbf{v}^{(k)}  \\
					\bm{\eta}^{(k+1)} -\bm{\eta}^{(k)} \\
				\end{bmatrix}.
			\end{aligned}
		\end{split}
	\end{equation}

	Since the mapping $F(\bm{\xi})$ is affine with a skew-symmetric matrix, it is monotonic \cite{He2012convergence}. Then we have the inequality 	
	\begin{equation}\label{eq:lemma_step_2}
		\begin{split}
			\begin{aligned}
				&\Xi(\bm{\zeta}^{(k+1)})  - \Xi(\bm{\zeta}^{\star})  + (\bm{\xi}^{(k+1)} - \bm{\xi}^{\star})^\T   F(\bm{\xi}^{(k+1)}) \\
				&\geq
				\Xi(\bm{\zeta}^{(k+1)})  - \Xi(\bm{\zeta}^{\star})  + (\bm{\xi}^{(k+1)} - \bm{\xi}^{\star})^\T   F(\bm{\xi}^\star) \geq \mathbf{0}.
			\end{aligned}
		\end{split}
	\end{equation}

	Meanwhile, using \eqref{eq:eta-primal}, the inequality can be written as
	\begin{equation}\label{eq:lemma_step_1}
		\begin{split}
			\begin{aligned}
				(\bm{\eta}^{(k)} -\bm{\eta}^{(k+1)} )^\T
				\left(- \begin{bmatrix}
					\mathbf{I}  \\
					\mathbf{G}\\
				\end{bmatrix} \right)
				(\mathbf{v}^{(k)} - \mathbf{v}^{(k+1)}) \geq  \mathbf{0}.
			\end{aligned} 
		\end{split}
	\end{equation}

	Combing \eqref{eq:lemma_step_2} and \eqref{eq:lemma_step_1}, we can derive
	\eqref{eq:subproblem_linear_gradient_3} as 
	\begin{equation}\label{eq:lemma_step_3}
		\begin{split}
			\begin{aligned}
				& \left(\begin{bmatrix}
					\mathbf{v}^{(k+1)}\\
					\bm{\eta}^{(k+1)}\\
				\end{bmatrix}  - \begin{bmatrix}
					\mathbf{v}^\star\\
					\bm{\eta}^\star\\
				\end{bmatrix}\right)^\T 
				\begin{bmatrix}
					\gamma \mathbf{I} + {\mathbf{G}}^\T {\mathbf{G}}  &  \mathbf{0}\\
					\mathbf{0}  &\frac{1}{\gamma}\mathbf{I} \\
				\end{bmatrix}
				\left(
				\begin{bmatrix}
					\mathbf{v}^{(k)}\\
					\bm{\eta}^{(k)}\\
				\end{bmatrix}  - \begin{bmatrix}
					\mathbf{v}^{(k+1)}\\
					\bm{\eta}^{(k+1)}\\
				\end{bmatrix}
				\right) \\ 
				&\geq
				\gamma  \begin{bmatrix}
					\bm{\xi}^{(k+1)} - \bm{\xi}^\star   \\
				\end{bmatrix}^\T \begin{bmatrix}
					\begin{bmatrix}
						\mathbf{\Phi} & \mathbf{0}\\
						\mathbf{0}  & \mathbf{I} 
					\end{bmatrix}^\T  \\
					-\begin{bmatrix}
						\mathbf{I}  \\
						\mathbf{G}
					\end{bmatrix}^\T  \\
					\mathbf{0} \\
				\end{bmatrix} \begin{bmatrix}
					\mathbf{I}  \\
					\mathbf{G}
				\end{bmatrix} \big( \mathbf{v}^{(k)} - \mathbf{v}^{(k+1)}\big)\\ 
				&\quad \geq (\bm{\eta}^{(k)} -\bm{\eta}^{(k+1)} )^\T 
				\left(- \begin{bmatrix}
					\mathbf{I}  \\
					\mathbf{G}\\
				\end{bmatrix}\right) \big(\mathbf{v}^{(k)} - \mathbf{v}^{(k+1)} \big)\geq \mathbf{0}. 
			\end{aligned} 
		\end{split}
	\end{equation}
	Let $\mathbf{\Omega} = \begin{bmatrix}
		\gamma \mathbf{I} + {\mathbf{G}}^\T {\mathbf{G}}  &  \mathbf{0}\\
		\mathbf{0}  &\frac{1}{\gamma}\mathbf{I} \\
	\end{bmatrix}$.	      
	We then conclude that
	\begin{equation}\label{eq:lemma_step_6}
		\begin{split}
			\begin{aligned}
				&\left\| \begin{bmatrix}
					\mathbf{v}^{(k)}\\
					\bm{\eta}^{(k)}\\
				\end{bmatrix}  - \begin{bmatrix}
					\mathbf{v}^\star\\
					\bm{\eta}^\star\\
				\end{bmatrix} \right\|_{\mathbf{\Omega}}^2  \\ 
				&=
				\left\| \begin{bmatrix}
					\mathbf{v}^{(k+1)}\\
					\bm{\eta}^{(k+1)}\\
				\end{bmatrix}  
				-  \begin{bmatrix}
					\mathbf{v}^\star\\
					\bm{\eta}^\star\\
				\end{bmatrix}
				\right\|_{\mathbf{\Omega}}^2 
				+ \left\| 
				\begin{bmatrix}
					\mathbf{v}^{(k)}\\
					\bm{\eta}^{(k)}\\
				\end{bmatrix}  - 
				\begin{bmatrix}
					\mathbf{v}^{(k+1)}\\
					\bm{\eta}^{(k+1)}\\
				\end{bmatrix}  
				\right\|_{\mathbf{\Omega}}^2
				\\
				&\,  
				+ 2 \left(  \begin{bmatrix}
					\mathbf{v}^{(k+1)}\\
					\bm{\eta}^{(k+1)}\\
				\end{bmatrix}  
				-  \begin{bmatrix}
					\mathbf{v}^\star\\
					\bm{\eta}^\star\\
				\end{bmatrix}
				\right)^\T \mathbf{\Omega} 
				\left(\begin{bmatrix}
					\mathbf{v}^{(k)}\\
					\bm{\eta}^{(k)}\\
				\end{bmatrix}  - 
				\begin{bmatrix}
					\mathbf{v}^{(k+1)}\\
					\bm{\eta}^{(k+1)}\\
				\end{bmatrix}  \right)\\
				&\geq \left\|\begin{bmatrix}
					\mathbf{v}^{(k+1)}\\
					\bm{\eta}^{(k+1)}\\
				\end{bmatrix}  - \begin{bmatrix}
					\mathbf{v}^\star\\
					\bm{\eta}^\star\\
				\end{bmatrix} \right\|_{\mathbf{\Omega}}^2
				+  	\left\|\begin{bmatrix}
					\mathbf{v}^{(k)}\\
					\bm{\eta}^{(k)}\\
				\end{bmatrix}  - 
				\begin{bmatrix}
					\mathbf{v}^{(k+1)}\\
					\bm{\eta}^{(k+1)} \\
				\end{bmatrix} \right\|_{\mathbf{\Omega}}^2.
			\end{aligned}
		\end{split}
	\end{equation}

	\section{Proof of Lemma 2}
	\label{pf:lemma-nonincreasing}
	To simplify the notation, we replace the $(k+1):$ iteration by the $+$:th iteration, and drop the iteration counter $k$ in this proof. Due to the strongly amenability, $s(\mathbf{x})$ is prox-regular with a positive constant $M$. 
	Now we compute 
	\begin{equation} \label{eq:lang_1}
		\begin{split}
			&  \mathcal{L}_{\gamma}(\mathbf{x},\mathbf{w},\mathbf{v};\bm{\eta}) - 
			\mathcal{L}_{\gamma}(\mathbf{x}^{(+)},\mathbf{w}^{(+)},\mathbf{v};\bm{\eta})  \\
			& =
			s(\mathbf{x}) - s(\mathbf{x}^{(+)}) + 
			\langle \overline{\bm{\eta}},  \mathbf{\Phi}\mathbf{x}^{(+)}- \mathbf{\Phi}\mathbf{x} \rangle \\
			&+
			\langle \gamma(\mathbf{\Phi} \mathbf{x}^{(+)}
			-\mathbf{d} - \mathbf{v} ),  \mathbf{\Phi}\mathbf{x}^{(+)}- \mathbf{\Phi}\mathbf{x} \rangle 
			+ \frac{\gamma}{2} \| \mathbf{\Phi}\mathbf{x}^{(+)}  - \mathbf{\Phi}\mathbf{x} \|^2 \\ 
			& \quad +
			g(\mathbf{w}) - g(\mathbf{w}^{(+)}) + 
			\langle \underline{\bm{\eta}},  \mathbf{w}^{(+)}-\mathbf{w} \rangle \\ 
			& \qquad +
			\langle\gamma ( \mathbf{w}^{(+)} - \mathbf{G} \mathbf{v} ),  \mathbf{w}^{(+)}- \mathbf{w}\rangle 
			+ \frac{\gamma}{2} \|\mathbf{w}^{(+)}  - \mathbf{w}\|^2
			\\
			&{>}  \frac{  \gamma  \delta_{+}(\mathbf{\Phi}^\T \mathbf{\Phi}) - M }{2} 
			\| \mathbf{x}^{(+)} - \mathbf{x} \|^2 + 
			\frac{\gamma}{2} \| \mathbf{w}^{(+)}- \mathbf{w} \|^2, 
		\end{split} 
	\end{equation} 
	where $\bm{\eta} = \operatorname{vec}( \overline{\bm{\eta}}, \underline{\bm{\eta}} )$.  We then have
	\begin{equation} \label{eq:q_case_a}
		\begin{split}
			\begin{aligned}
				& \mathcal{L}_{\gamma}(\mathbf{x}^{(+)},\mathbf{w}^{(+)},\mathbf{v}^{(+)};\bm{\eta}^{(+)})  -\mathcal{L}_{\gamma}(\mathbf{x},\mathbf{w},\mathbf{v};\bm{\eta}) 
				<   \frac{1}{\gamma} \| \bm{\eta}^{(+)}- \bm{\eta} \|^2  \\
				& +\frac{M - \gamma  \delta_{+}(\mathbf{\Phi}^\T \mathbf{\Phi}) }{2}
				\| \mathbf{x}^{(+)}- \mathbf{x} \|^2  
				+\frac{\gamma}{2} \| \mathbf{w}^{(+)} -  \mathbf{w} \|^2,
			\end{aligned}
		\end{split} 
	\end{equation}
	which will be nonnegative provided when $\gamma > \frac{M}{\delta_{+}(\mathbf{\Phi}^\T \mathbf{\Phi}) }> 0 $ is satisfied. In particular, when $\mathbf{\Phi} = \mathbf{I}$,  $ \delta_{+}(\mathbf{\Phi}^\T \mathbf{\Phi})  = 1$. 
	
	In our case, $\mathcal{L}_{\gamma}(\mathbf{x}^{(k)}, \mathbf{w}^{(k)},\mathbf{v}^{(k)}; \bm{\eta}^{(k)})$ is upper bounded by $\mathcal{L}_{\gamma}(\mathbf{x}^{(0)}, \mathbf{w}^{(0)},\mathbf{v}^{(0)}; \bm{\eta}^{(0)})$, and is also lower bounded by 
	$  
	\mathcal{L}_{\gamma}(\mathbf{x}^{(k)}, \mathbf{w}^{(k)}, \mathbf{v}^{(k)}; \bm{\eta}^{(k)}) \geq s(\mathbf{x}^{(k)})  +  \sum_{t = 1}^{T}\sum_{g = 1}^{N_g} \mu  \| \mathbf{w}_{g,t} \|_2. 
	$
	Thus, we get the conclusion.

	\section{Proof of Lemma 3}
	\label{pf:lemma-LM}
	
	We use the smallest non-zero eigenvalue of $ \mathbf{\Phi}^\T  \mathbf{\Phi} $ and $ \mathbf{S}^{-1} $ to yield the inequality
	\begin{equation}
		\begin{split}
			\begin{aligned}
				\left\|\mathbf{J}_{\theta}^\top  \mathbf{J}_{\theta}(\mathbf{x}^{(i)})  \right\|
				\geq \max \left\{\gamma \delta_{+}(\mathbf{\Phi}^\top \mathbf{\Phi}), \lambda^{(i)} \delta_{+}([\mathbf{S}^{(i)}]^{-1}) \right\}, 
			\end{aligned}
		\end{split}
	\end{equation}
	where $\mathbf{J}_{\theta} = \begin{bmatrix} \mathbf{R}^{-\frac{1}{2}} \mathbf{J}_{h}(\mathbf{x} ) &
		\mathbf{Q}^{-\frac{1}{2}} \mathbf{J}_{a}(\mathbf{x} ) &
		\gamma^{\frac{1}{2}} \mathbf{\Phi}&
		\lambda^{\frac{1}{2}} \mathbf{S}^{-\frac{1}{2}}
	\end{bmatrix}^\T$. 
	We then have
	\begin{equation}
		\label{eq:taylorbased_equality_lm}
		\begin{aligned}
			&  \left\| \mathbf{x}^{(i+1)} - \mathbf{x}^\star \right\|  
			\leq \frac{M}{2}  \left\| [\mathbf{J}_{\theta}^\T \mathbf{J}_{\theta}(\mathbf{x}^{(i )}) ] ^{-1}  \right\| 
			\left\| \mathbf{x}^{(i)} - \mathbf{x}^\star  \right\|^2  \\
			& \quad + 
			\left\|  [\mathbf{J}_{\theta}^\T \mathbf{J}_{\theta}(\mathbf{x}^{(i )}) ] ^{-1} \mathbf{H}_{\theta} (\mathbf{x}^{(i )}) \right\|
			\left\| \mathbf{x}^{(i)} - \mathbf{x}^\star \right\|.
		\end{aligned}
	\end{equation}
	When $\|\mathbf{H}_{\theta}(\mathbf{x}) \| \le \kappa $ and $\kappa \to 0$, the convergence is quadratic. The linear convergence can be established when the inequality
	$\left\|  [\mathbf{J}_{\theta}^\T \mathbf{J}_{\theta}(\mathbf{x}^{(i )}) ]^{-1} \mathbf{H}_{\theta} (\mathbf{x}^{(i )}) \right\| 
	\leq  {\kappa}/{\max \left\{\gamma \delta_{+}(\mathbf{\Phi}^\top \mathbf{\Phi}), \lambda^{(i)} \delta_{+}([\mathbf{S}^{(i)}]^{-1}) \right\}} < 1,
	$
	is satisfied.

	\ifCLASSOPTIONcaptionsoff
	\newpage
	\fi

	\bibliographystyle{IEEEtran}
	\bibliography{refs,refs2}

\end{document}